\documentclass[12pt]{amsart}

\usepackage{bold-extra}
\usepackage[normalem]{ulem}

\usepackage{microtype}
\usepackage{amssymb}
\usepackage{tikz}
\usepackage{tikz-cd}
\usepackage{bigints}
\usepackage{amsthm}
\usepackage{float} 
\usepackage{mathtools}
\usepackage[font=small,labelfont=bf]{caption}
\usepackage{xparse}
\usepackage{subcaption}
\usepackage{suffix}
\usepackage{scalerel,bbm,verbatim,}

\definecolor{darkblue}{rgb}{0,0,0.5}
\definecolor{red}{rgb}{1,0,0}
\definecolor{darkgreen}{rgb}{0,0.5,0}
\definecolor{darkred}{rgb}{0.75,0,0}

\usepackage[colorlinks=true, linkcolor=darkblue, urlcolor=darkblue, citecolor=darkgreen]{hyperref}

\usepackage{thmtools}
\usepackage[nameinlink, capitalise]{cleveref}

\newcommand{\Z}{\mathbb{Z}}

\newcommand{\SAT}{3\mathrm{SAT}}

\newcommand{\SLC}{\mathrm{SLC}}

\numberwithin{equation}{section}
\newtheorem{theorem}{Theorem}

\newtheorem{lemma}[theorem]{Lemma}
\newtheorem{corollary}[theorem]{Corollary}
\newtheorem{proposition}[theorem]{Proposition}

\newtheorem*{claim*}{Claim}

\newtheorem*{conjecture*}{Conjecture}
\newtheorem*{definition*}{Definition}
\newtheorem*{theorem*}{Theorem}
\newtheorem*{remark*}{Remark}

\newtheorem*{lemma*}{Lemma}
\newtheorem*{question*}{Question}
\newtheorem{definition}[theorem]{Definition}

\newtheorem*{example*}{Example}
\newtheorem{remark}[theorem]{Remark}

\newcommand{\R}{\mathbb{R}}
\newcommand{\C}{\mathbb{C}}
\newcommand{\N}{\mathbb{N}}

\newcommand{\tr}{\operatorname{tr}}

\newcommand{\expect}{\mathbb{E}}

\newcommand{\NP}{\mathrm{NP}} %
\newcommand{\MIP}{\mathrm{MIP}} %
\newcommand{\RE}{\mathrm{RE}} %
\newcommand{\BCSMIP}{\mathrm{BCS\text{-}MIP}}

\newcommand{\poly}{\mathrm{poly}}
\newcommand{\polylog}{\mathrm{polylog}}

\let\epsilon=\varepsilon 

\newcommand{\Paren}[1]{\left(#1\right)}

\usepackage{pifont}

\newcommand{\mc}[1]{\mathcal{#1}}

\newcommand{\mbb}[1]{\mathbb{#1}}
\newcommand{\mbbm}[1]{\mathbbm{#1}}

\renewcommand{\C}{\mathbb{C}}
\newcommand{\Id}{\mathbb{I}}
\newcommand{\SuccinctCSP}{\mathrm{SuccinctCSP}}

\DeclareMathOperator{\defect}{def}

\DeclareMathOperator{\supp}{supp}

\DeclarePairedDelimiter\parens{\lparen}{\rparen}
\DeclarePairedDelimiter\squ{[}{]}

\DeclarePairedDelimiter\abs{\lvert}{\rvert}
\DeclarePairedDelimiter\norm{\|}{\|}
\DeclarePairedDelimiter\floor{\lfloor}{\rfloor}
\DeclarePairedDelimiter\ceil{\lceil}{\rceil}
\DeclarePairedDelimiter\bra{\langle}{\rvert}
\DeclarePairedDelimiter\ket{\lvert}{\rangle}

\DeclareMathOperator*{\bigast}{\scalerel*{\ast}{\sum}}

\DeclareDocumentCommand\hsq{s o m}{
	\IfBooleanTF{#1}{
        \left|{#3}\right|^2
	}{
        \IfNoValueTF{#2}{
            |{#3}|^2
        }{
            {#2|}{#3}{#2|}^2
        }
    }
}

\DeclareDocumentCommand\ketbra{s m g}{
	\IfBooleanTF{#1}{
		\IfNoValueTF{#3}{
			\vphantom{#2}\left\lvert{#2}\middle\rangle\!\middle\langle{#2}\right\rvert
		}{
			\vphantom{#2#3}\left\lvert{#2}\middle\rangle\!\middle\langle{#3}\right\rvert
		}
	}{
		\IfNoValueTF{#3}{
			\vphantom{#2}\left\lvert\smash{#2}\middle\rangle\!\middle\langle\smash{#2}\right\rvert
		}{
			\vphantom{#2#3}\left\lvert\smash{#2}\middle\rangle\!\middle\langle\smash{#3}\right\rvert
		}
	}
}

\DeclareDocumentCommand\braket{s m g g}{
	\IfBooleanTF{#1}{
		\IfNoValueTF{#3}{
			\vphantom{#2}\left\langle{#2}\middle\vert{#2}\right\rangle
		}{
			\IfNoValueTF{#4}{
				\vphantom{#2#3}\left\langle{#2}\middle\vert{#3}\right\rangle
			}{
				\vphantom{#2#3#4}\left\langle{#2}\middle\rvert{#3}\middle\rvert{#4}\right\rangle
			}
		}
	}{
		\IfNoValueTF{#3}{
			\vphantom{#2}\left\langle\smash{#2}\middle\vert\smash{#2}\right\rangle
		}{
			\IfNoValueTF{#4}{
				\vphantom{#2#3}\left\langle\smash{#2}\middle\vert\smash{#3}\right\rangle
			}{
				\vphantom{#2#3#4}\left\langle\smash{#2}\middle\rvert\smash{#3}\middle\rvert\smash{#4}\right\rangle
			}
		}
	}
}

\DeclareDocumentCommand\gen{s m g}{
	\IfBooleanTF{#1}{
		\IfNoValueTF{#3}{
			\vphantom{#2}\left\langle{#2}:{#2}\right\rangle
		}{
            \vphantom{#2#3}\left\langle{#2}:{#3}\right\rangle
		}
	}{
		\IfNoValueTF{#3}{
			\vphantom{#2}\left\langle\smash{#2}:\smash{#2}\right\rangle
		}{
			\vphantom{#2#3}\left\langle\smash{#2}:\smash{#3}\right\rangle
		}
	}
}

\DeclareDocumentCommand\set{s m g}{
	\IfBooleanTF{#1}{
		\IfNoValueTF{#3}{
			\vphantom{#2}\left\{{#2}\right\}
		}{
			\vphantom{#2#3}\left\{{#2}\mid{#3}\right\}
		}
	}{
		\IfNoValueTF{#3}{
			\vphantom{#2}\left\{\smash{#2}\right\}
		}{
			\vphantom{#2#3}\left\{\smash{#2}\mid\smash{#3}\right\}
		}
	}
}

\newcommand{\cnote}[1]{\textcolor{purple}{\small (Connor: #1)}}

\usepackage{etoolbox}
\usepackage{appendix}
\AtBeginEnvironment{appendices}{\crefalias{section}{appendix}}
\usepackage[top=1in,bottom=1in,left=1in,right=1in]{geometry}
\usepackage[foot]{amsaddr}
\title[Quantum smooth label cover is undecidable]{The quantum smooth label cover problem \\ is undecidable}

\author[E.~Culf]{Eric Culf$^{\mathsection *}$}
\author[K.~Mastel]{Kieran Mastel$^{\dagger *}$}
\author[C.~Paddock]{Connor Paddock$^{\ddagger}$}
\author[T.~Spirig]{Taro Spirig$^{\#}$}

\address{$^\mathsection$Department of Applied Mathematics, University of Waterloo}
\address{$^\dagger$Department of Pure Mathematics, University of Waterloo}
\address{$^\ddagger$Department of Mathematics and Statistics, University of Ottawa}
\address{$^*$Institute for Quantum Computing, University of Waterloo}
\address{$^{\#}$QMATH, Department of Mathematical Sciences, University of Copenhagen}

\email{eculf@uwaterloo.ca}

\email{kmastel@uottawa.ca}

\email{cpaulpad@uottawa.ca}

\email{tasp@math.ku.dk}

\begin{document}

\begin{abstract}
We show that the quantum smooth label cover problem is undecidable and $\RE$-hard. This sharply contrasts the quantum unique label cover problem, which can be decided efficiently by a result of Kempe, Regev, and Toner (FOCS'08). On the other hand, our result aligns with the $\RE$-hardness of the quantum label cover problem, which follows from the celebrated $\MIP^*=\RE$ result of Ji, Natarajan, Vidick, Wright, and Yuen (ACM'21). Additionally, we show that the quantum oracularized smooth label cover problem is $\RE$-hard. Our second result fits with the alternative quantum unique games conjecture recently proposed by Mousavi and Spirig (ITCS'25) on the $\RE$-hardness of the quantum oracularized unique label cover problem. Our proof techniques include a quantum version of Feige's reduction from 3SAT to 3SAT5 (STOC'96) for $\BCSMIP^*$-protocols, which may be of independent interest.
\end{abstract}

\maketitle

\section{Introduction}

The unique games conjecture is one of the most important open problems in theoretical computer science. Provided that P $\neq\NP$, it imposes strict limitations on the existence of efficient algorithms to approximate the value of many NP-hard problems. As an important example, under the unique games conjecture the Goemans and Williamson polynomial-time approximation algorithm for MAXCUT is optimal~\cite{goemansmaxcut,khot_original}. Despite enormous attention and effort \cite{prasad,khot-regev-vertex-cover,KhotGrassmann1,DinurTwotoOne}, the conjecture remains unresolved. The conjecture, attributed to Khot \cite{khot_original}, informally states that for any $\epsilon,\delta>0$, given a large enough alphabet, the unique label cover problem (\textit{i.e.}~a unique game), it is $\NP$-hard to decide if the classical value is at least $1-\epsilon$ or no more than $\delta$.

Surprisingly, Kempe, Regev, and Toner showed that the quantum analogue of the unique games conjecture is false. In particular, they proved that solutions of a semi-definite program (SDP) relaxation of any unique game can be rounded to entangled strategies that well-approximate the quantum value of the game \cite{kempe}. At the time, this furthered the idea that the quantum value might be easier to compute than the classical value of many two-player games. However, in light of the recent $\MIP^*=\RE$ result \cite{ji_mip_re}, which implies that even approximating the quantum value of two-player (nonlocal) games is $\RE$-complete, the fact that an efficient approximation algorithm exists for the quantum value of unique games is even more peculiar.

In the classical setting, a source of evidence for the unique games conjecture is the $\NP$-hardness of the label cover and smooth label cover problems \cite{feige,khot2,guruswami-bypassing}. Hardness of the label cover problem --- a less structured variant of unique label cover --- follows from the PCP theorem \cite{arora1998proof,arora1998probabilistic}. While smooth label cover is more akin to unique label cover, the smooth variant is not an interpolation between the other two. Nonetheless, the $\NP$-hardness of smooth label cover has been a powerful tool in its own right, used to establish hardness of approximation results for many problems in $\NP$ \cite{guruswami-bypassing}, even in cases where the unique games conjecture was thought to be necessary.

In the quantum setting of interactive proofs with entangled provers, the story is quite different. As mentioned, the quantum value of a unique game can be approximated in polynomial-time by the algorithm described in \cite{kempe}. On the other hand, Mousavi and Spirig recently pointed out that, by the $\MIP^*=\RE$ result, the analogous quantum label cover problem is \emph{undecidable} and RE-hard \cite{mousavispirig24}. This leaves us with a striking dichotomy, and begs the natural question of whether the quantum smooth label cover problem is easy, undecidable, or somewhere in between.

\subsection{Results Overview}

In more detail, a two-player game is akin to a two-prover interactive proof system. Here, non-communicating provers (or players) interact with an efficient verifier (or referee) through a single round of question and answer. The players are cooperating and their goal is to win the game by satisfying the rule predicate, a condition checked by the referee. In the quantum setting the players have access to the additional resource of entangled quantum states, and the ability to measure them locally between receiving their questions and returning their answers.

A smooth label cover instance consists of a bipartite graph $G$ with disjoint vertex sets $U\cup V$, edges $E$, and an alphabet of size $n$. Given an instance, the goal is to assign labels to each of the vertices, maximizing the number of valid assignments; the validity of a label assignment to a vertex is determined by constraints imposed by the edges of the graph. Smoothness is a certain ``global-to-local'' property on the constraints, which ensures the assignments satisfying any near-optimal number of constraints have specific local features. It follows that each smooth label cover problem is an instance of a 2-ary constraint satisfaction problem ($2$-$\mathrm{CSP}$), with variables (of degree $n$) corresponding to the vertices, and constraints coming from the edges.

An instance of smooth label cover game can be played as a constraint system nonlocal game $\mc{G}_{\mathrm{SLC}(n)}$. In the game, the referee samples a constraint $e_{uv}$ according to some distribution, and sends one variable $u$ to the first player, and the other variable $v$ to the second player. The first player responds with an assignment to variable $u$, and the second player with an assignment to variable $v$. Finally, the referee checks that the assignments to $u$ and $v$ satisfy the constraint $e_{uv}$. The completeness of the game is clear: if they have an assignment which satisfies a large number of constraints, then they are likely to win the game. The notion of soundness is similar: players with a good strategy will be able to construct a good assignment.

The optimal classical value of the smooth label cover game $\mc{G}_{\mathrm{SLC}(n)}$, denoted by $\omega_c(\mc{G}_{\mathrm{SLC}(n)})$, will be directly related to the optimal fraction of satisfied vertices in the smooth label cover instance. That is, we can recall the classical hardness of the smooth label cover problem as follows. For all $0<s<1$, there is an $n$ large enough, such that the following decision problem is $\NP$-hard: given $\mc{G}_{\mathrm{SLC}(n)}$ decide if $\omega_c(\mc{G}_{\mathrm{SLC}(n)})=1$ or $\omega_c(\mc{G}_{\mathrm{SLC}(n)})<s$, promised that one holds\footnote{Unlike for unique label cover, smooth label cover with perfect completeness is not necessarily easy.}. By allowing quantum strategies for the corresponding 2-CSP game, the quantum smooth label cover problem is the following: given $\mc{G}_{\mathrm{SLC}(n)}$ and $0<s<1$, decide if $\omega_q(\mc{G}_{\mathrm{SLC}(n)})=1$ or $\omega_q(\mc{G}_{\mathrm{SLC}(n)})<s$, promised that one holds. Where the quantum value $\omega_q(\mc{G}_{\mathrm{SLC}(n)})$ is the supremum over all quantum strategies for~$\mc{G}_{\mathrm{SLC}(n)}$. In this work, we establish the following result.

\begin{theorem*}[Informal]
    For any soundness parameter, there exists a sufficiently large alphabet such that the quantum smooth label cover problem is $\RE$-hard.
\end{theorem*}

With this result, it seems as though quantum unique label cover problem is the most strange, effectively switching from hard to easy. See \cref{thm:QSmoothLC} for the formal statement.

One can take the perspective that the $\MIP^*=\RE$ result is a quantum version of the classical PCP theorem\footnote{Although, there are some good reasons to not compare these statements directly, see \cite{anand-games-pcp}.} for succinctly presented constraint satisfaction problems (CSPs). Unlike with the classical $\MIP=\mathrm{NEXP}$ result \cite{babai}, computable reductions between problems in $\RE$ are not sensitive to whether we have a compact presentation of the CSP. In other words, if a succinctly presented CSP is undecidable, then so is the non-succinct version of the problem. In this sense, the $\MIP^*=\RE$ result of \cite{ji_mip_re} puts significant limitations on the existence of \emph{any} approximation algorithms for the quantum value of nonlocal games. Still, understanding which games admit approximation algorithms for the quantum value is a very interesting problem. Here, little is known compared to the classical case, beyond the setting of XOR games \cite{cleve2008perfect,wehner2006tsirelson,cleve}, and some other partial results \cite{CMS23}.

Mousavi and Spirig point out that the $\MIP^*=\RE$ result naturally implies that the \emph{quantum oracularized} label cover problem is also $\RE$-hard. Oracularizable quantum strategies are a restricted class of quantum synchronous strategies which have convenient properties in the context of reductions between $\MIP^*$ protocols for CSPs. In particular, they are necessary for the answer reduction portion of the $\MIP^*=\RE$ proof \cite{ji_mip_re,dong2023computational}. This led Mousavi and Spirig to provide an alternative quantum unique games conjecture, based on the quantum oracularized value, to which the approximation algorithm of \cite{kempe} does not readily apply. Assuming their conjecture, they established a hardness of approximation result for the quantum oracularized value for families of CSPs from $\text{2-}\mathrm{LIN}$ and $\mathrm{MAXCUT}$ which resemble those established in the classical case under the unique games conjecture by Khot and others \cite{khot_original,kkmo}.

Our results for quantum smooth label cover aligns with their quantum unique games conjecture. Moreover, we were able to extend our result to the oracularized setting. Accordingly, our second main result extends the hardness of quantum smooth label cover to the quantum oracularized value. In other words, we show that the problem of deciding if $\omega_{qo}(\mc{G}_{\mathrm{SLC}(n)})=1$ or $\omega_{qo}(\mc{G}_{\mathrm{SLC}(n)})<s$, is RE-hard, where $\omega_{qo}(\mc{G}_{\mathrm{SLC}(n)})$ is the supremum over all quantum oracularizable strategies. For concreteness, one can consider the smooth label cover game, where one of the players must reply with an assignment to both vertices $u$ and $v$, and the other player obtains only one of $u$ or $v$ chosen at random by the verifier, and replies with an assignment to that vertex.

\begin{theorem*}[Informal]
    For any soundness parameter, there is a sufficiently large alphabet such that the quantum oracularized smooth label cover problem is $\RE$-hard.
\end{theorem*}

The formal result can be found in \cref{thm:QOSmoothLC}.

Oracularizable strategies for 2-CSP protocols like smooth label cover are more like strategies for two-player CSP protocols where the constraints contain more than two variables. In two-player protocols for CSPs with constraints on more than two variables, at least one of the players needs to respond with assignments to more than one variable at a time. Therefore, that player's measurements for the variables in a constraint must commute with one another. On the other hand, in a non-oracularizable strategy for a $2$-CSP protocol, each player is sent one variable, so the player's measurements for these variables need not commute, since a single player only measures one variable. These sound noncommutative protocols only exist for 2-CSPs, making this \emph{constraint-noncommutativity} a behaviour unique to entangled 2-CSP protocols. In an oracularizable strategy for a $2$-CSP game, measurements for the variables in each constraint commute. This is exactly like quantum strategies for CSP protocols with constraints on more than 2 variables, such as the constraint-constraint and constraint-variable games, see for example \cite{BCSTensor,ji2013binary,atserias2019generalized,paddock2025satisfiability,culfmastel24}.

The existence of non-oracularizable strategies for 2-CSPs is one of the reasons why the authors of \cite{mousavispirig24} suggest a unique games conjecture based on the hardness of the quantum oracularized value, as it fits with the general theory of entangled CSP protocols of higher arity. In particular, the entangled strategies produced by efficient procedure of \cite{kempe} are not oracularizable and therefore the RE-hardness of the quantum oracularized label cover problem is entirely plausible.

\subsection{Technical Outline}

Several of the proof ideas in this work are inspired by classical reductions between decision problems. Establishing completeness of a classical reduction in the quantum case is often more straightforward, and it is establishing the quantum soundness which often presents a new challenge. To overcome this obstacle in many cases, we make use of the powerful algebraic formalism for handling soundness arguments for $\MIP^*$ protocols for CSPs that was recently introduced in \cite{MS24} and expanded on in \cite{culfmastel24}. The weighted algebra formalism examines how reductions change the quantum \emph{synchronous} value. One of the main observations used throughout this work (as well as, for instance \cite{MS24,culfmastel24,fu2025succinct}) is that we can reduce questions about the hardness of the quantum value to the hardness of the quantum synchronous value.

To prove our first main result, we give a reduction from the halting problem to the quantum smooth label cover problem. Our starting point is the $\RE$-hardness for $\SAT^*$ established by Culf and Mastel \cite{culfmastel24}, where $\SAT^*$ denotes the $\MIP^*$ protocol for $\SAT$. Next, we exhibit a reduction from this RE-hard problem to $\SAT5^*$, the corresponding $\MIP^*$ protocol where every variable appears in (exactly) $5$ constraints. From here, we employ the classical construction of Guruswami, Raghavendra, Saket, and Wu \cite{guruswami-bypassing}, building on the work of Khot \cite{khot2}, to obtain a smooth label cover instance. One of our technical contributions is ensuring the completeness and soundness of the Guruswami, Raghavendra, Saket, and Wu construction in the quantum case. This involves considering a variant of the $(J,R)$-dummy variable game (see \cref{def:ord_JR_game}) where the questions and answers are ordered, since completeness of parallel repetition in the quantum setting depends on the players knowing which round of the repeated game they are playing. However, the primary technical contribution of this work is establishing the reduction from $\SAT^*$ to $\SAT5^*$.

Our approach is inspired by the classical reductions from $\SAT$ to $\SAT B$ (where every variable appears in $B$ constraints) by Papadimitriou and Yannakakis \cite{PY88} and Arora et al.~ \cite{ALMSS98}, and then from $\SAT B$ to $\SAT5$ by Feige \cite{Feige98threshold}. The idea involves modifying the $\SAT$ system by first labelling each variable by its constraint and then adding equality constraints between these labelled variables. One can think of the labelled variables and equality constraints between them as the vertices and edges of a graph. Unfortunately, this reduction does not decrease the degree of the $\SAT$ instance. To overcome this issue, as in the classical reductions, we modify the construction by adding equality constraints only along the edges of an expander graph. Unlike in the classical setting, analysing the soundness of the expander construction is less straightforward, however we are able to prove that it remains sound in the quantum setting. The key observation is that we can employ a nice result from \cite{Ji2021quantum} (\cref{lem:key_expand}) to obtain a constant soundness drop-off in the quantum case as well.

The soundness proof of the expander graph construction hinges on a certain uniformity assumption for the question distribution of the initial $\SAT^*$ protocol. Hence, to complete the proof of $\RE$-hardness for $\SAT5^*$, we need to establish RE-hardness for such uniform $\SAT^*$ protocols. Unfortunately, this is not clear a priori via polynomial-time reductions. Nevertheless, we modify the halting problem so that the required uniformity condition is baked in. This modified halting problem is $\RE$-hard by an exponential-time reduction from the halting problem. Then, the $\RE$-hardness of uniform $\SAT^*$ follows by a polynomial-time reduction from the modified halting problem.

Our second main result, establishing undecidability in the oracularized case, follows by a modification of the $\SAT5^*$ protocol described above. This modification results in a $\SAT10^*$ protocol, in which the quantum assignments are what we call \emph{2-oracularizable}. That is, pairs of observables corresponding to variables, say $x$ and $y$, will commute not only if they are in the same constraint, but also whenever another variable $z$ appears in separate constraints with each of $x$ and $y$, respectively. The 2-oracularizable property of the quantum assignments to the $\SAT10^*$ protocol ensures that in the construction of the smooth label cover game $\mc{G}_{\mathrm{SLC}(n)}$, the quantum strategies will be fully oracularizable. The construction is such that quantum assignments to the $\SAT5^*$ protocol will give assignments to the $\SAT10^*$ protocols. Lastly, the soundness of the construction follows by observing that the entire modification is an instance of a ``constraint subdivision'' transformation, and therefore quantum soundness follows by a result of \cite{culfmastel24}, using the weighted algebras formalism.

Although runtime is not typically considered in the theory of computability, we recall that the reduction from the halting problem to $\BCSMIP^*$ with polylog-sized questions (and constant-sized answers) in \cite{fu2025succinct} runs in polynomial time. With this in mind, we show (in \cref{sec:uniformmarginals}) that there is a polynomial time reduction from the halting problem to a $\BCSMIP^*$ protocol with the aforementioned uniformity condition. However, since this protocol is succinctly presented, the expander graph construction takes quasi-polynomial time. Notably, if we had a version of the $\MIP^* = \RE$ theorem with log-sized questions (and constant-sized answers), it would give us a polynomial-time reduction from the halting problem to the quantum smooth label cover problem\footnote{Further reducing the question size in the $\MIP^*$ protocol for the halting problem is not believed to be the difficult part of the quantum games PCP conjecture \cite{anand-games-pcp}.}.

\section{Decision Problems for Constraint Satisfaction Games}

\subsection{Nonlocal Games}\label{sec:nonlocalgames}

A two-player \textbf{nonlocal game} $\mc{G} = (I,\{O_i\}_{i\in I},\pi,V)$ consists of a finite set of questions $I$, a collection of finite answer sets $\{O_i\}_{i\in I}$, a probability distribution $\pi$ on $I\times I$, and a family of functions $V(\cdot,\cdot|i,j):O_i\times O_j\rightarrow \{0,1\}$ for $(i,j)\in I\times I$. In the game, the players, often called Alice and Bob, receive questions $i$ and $j$, respectively, from $I$ with probability $\pi(i,j)$, and, without communicating with one another, reply with answers $a\in O_i$ and $b\in O_j$, respectively. They win if $V(a,b|i,j) = 1$ and lose otherwise. We have assumed without loss of generality that the players' question and answer sets are the same. 

Although they cannot communicate during the game, the players can coordinate a strategy beforehand. The players' strategy in a nonlocal game is described by the probability distributions $p(a,b|i,j)$ of their answers conditioned on the question pair they receive. Such distributions are known as \textbf{correlations} for the game $\mc{G}$. Depending on the resources available to the players, they can only employ certain correlations. A correlation $p$ is \textbf{quantum} if there are
\begin{enumerate}
    \item finite-dimensional Hilbert spaces $\mc{H}_A$ and $\mc{H}_B$,
    \item a projective measurement $\{M_a^i\}_{a\in O_i}$ on $\mc{H}_A$ for every $i \in I$,
    \item a projective measurement $\{N_a^i\}_{a \in O_i}$ on $\mc{H}_B$ for every $i \in I$, and 
    \item a state $\ket{v}\in \mc{H}_A\otimes \mc{H}_B$
\end{enumerate}
such that $p(a,b|i,j) =\bra{v}M^i_a\otimes N_b^j\ket{v}$ for all $i,j \in I$, $a \in O_i$, and $b \in O_j$. The collection $(\mc{H}_A, \mc{H}_B, \{M_a^i\},\{N_a^j\}, \ket{v})$ is called a \textbf{quantum strategy}. Quantum strategies capture the scenario where the players share some bipartite quantum state and sample their answers via measurements of that state. If the players instead only have access to classical resources such as shared randomness, their correlation and strategy are called \textbf{classical}. A correlation is \textbf{commuting operator} if there exists
\begin{enumerate}
    \item a Hilbert space $\mc{H}$,
    \item projective measurements $\{M_a^i\}_{a\in O_i}$ and $\{N_a^i\}_{a \in O_i}$ on $\mc{H}$ for every $i \in I$, and
    \item a state $\ket{v}\in \mc{H}$
\end{enumerate}
such that $M_a^iN_b^j = N_b^jM_a^i$ and $p(a,b|i,j) =\bra{v}M^i_a N_b^j\ket{v}$ for all $i,j \in I$, $a \in O_i$, and $b \in O_j$. The collection $(\mc{H},\{M_a^i\},\{N_a^j\},\ket{v})$ is called a \textbf{commuting operator strategy}. The quantum correlations are contained in the commuting operator correlations. If a commuting operator correlation has a commuting operator strategy on a finite dimensional Hilbert space, then it is also a quantum correlation. 

The \textbf{winning probability} of a correlation $p$ in a nonlocal game $\mc{G} = (I,\{O_i\},\pi,V)$ is 
\begin{equation*}
    \omega(\mc{G};p):= \sum_{i,j\in I}\sum_{a\in O_i, b\in O_j}\pi(i,j)V(a,b|i,j)p(a,b|i,j).
\end{equation*}
Sometimes we denote a strategy by $\mc{S}$, and write $\omega(\mc{G};S)$ for the winning probability. Let $C_q$ be the set of quantum correlations, then the \textbf{quantum value} of $\mc{G}$ is $\omega_q(\mc{G}) := \sup_{p\in C_q}(\mc{G};p)$. Similarly, let $C_{qc}$ be the set of commuting operator correlations, then the \textbf{commuting operator value} of $\mc{G}$ is $\omega_{qc}(\mc{G}) := \sup_{p\in C_{qc}}(\mc{G};p)$. Lastly, if $C_c$ is the set of quantum correlations, then the \textbf{classical value} of $G$ is $\omega_c(\mc{G}) := \sup_{p\in C_c}(\mc{G};p)$. A correlation $p$ for $\mc{G}$ and its corresponding strategy are called \textbf{perfect} if $\omega(\mc{G};p) = 1$, an $\epsilon$-perfect if $\omega(\mc{G};p)>1-\epsilon$. Since the set of quantum correlations is not closed by \cite{slofstra_set_of_quantum_correlations}, we say that a game $\mc{G}$ with $\omega_q(\mc{G}) = 1$ has a perfect \textbf{quantum approximate} strategy.

A nonlocal game $\mc{G} = (I,\{O_i\},\pi,V)$ is \textbf{synchronous} if $V(a,b|i,i) = 0$ for all $i \in I$ and $a \neq b \in O_i$. A correlation $p$ is \textbf{synchronous} if $p(a,b|i,i) = 0$ for all $i \in I$ and $a \neq b \in O_i$. The set of synchronous classical, quantum and commuting operator correlations are denoted $C_c^s$, $C_q^s$, and $C_{qc}^s$, respectively. We define the synchronous quantum value $\omega_q^s$ and synchronous commuting operator value $\omega_{qc}^s$ analogously to $\omega_{q}$ and $\omega_{qc}$ by replacing $C_q$ with $C_q^s$ and $C_{qc}$ with $C_{qc}^s$ respectively. If $p$ is a synchronous commuting operator correlation, then there is a single projective measurement $\{M^i_a\}_{a\in O_i}$ on $\mc{H}$ for each $i\in I$, and the state $\ket{v}$ is tracial, in the sense that $\bra{v}\alpha\beta\ket{v} = \bra{v}\beta\alpha\ket{v}$ for all $\alpha$ and $\beta$ in the $\ast$-algebra generated by the operators $M_a^i$, $i\in I$, $a \in O_i$. Moreover, the correlation can be written as $p(a,b|i,j) = \bra{v}M_a^iM_b^j\ket{v}$ for all $i,j\in I$, $a\in O_i$, and $b \in O_j$. The collection $(\mc{H},\{M_a^i\},\ket{v})$ is called a \textbf{synchronous commuting operator strategy}. If the Hilbert space $\mc{H}$ is finite dimensional, then the strategy is also a \textbf{synchronous quantum strategy}. A synchronous strategy is called \textbf{oracularizable} if $M_a^iM_b^j = M_b^jM_a^i$ for all $i,j \in I$, $a \in O_i$, and $b \in O_j$ with $\pi(i,j)>0$. Similarly, we call a correlation \textbf{oracularizable} if it admits an oracularizable strategy. Let $C_{qo}$ be the set of oracularizable quantum correlations, then the \textbf{quantum oracularizable value} of $\mc{G}$ is $\omega_{qo}(\mc{G}) := \sup_{p\in C_{qo}}(\mc{G};p)$. In some previous work on 2-CSP games, namely \cite{mousavispirig24}, the quantum synchronous value of a synchronous game has been called the non-commutative value, and the oracularized synchronous quantum value has been called the quantum value. We use terminology that is consistent with the broader nonlocal games literature.

Every quantum correlation that is approximately synchronous, in the sense that $p(a,b|i,i)$ $\approx 0$ for all $i \in I$ and $a\neq b \in O_i$, is close to an exactly synchronous quantum correlation \cite{Vidick_2022}. This is also true for commuting operator correlations \cite{lin2024tracialembeddable, marrakchi2023synchronous}. As a result, the synchronous quantum and synchronous commuting operator values of a synchronous game are related by a polynomial to the quantum and commuting operator values, respectively.

For a synchronous nonlocal game $\mc{G}$,
\begin{equation}\label{eq:game-values}
    \omega_c(\mc{G}) \leq \omega_{qo}(\mc{G})\leq \omega_q^s(\mc{G}) \leq \omega_{qc}^s(\mc{G}).
\end{equation}

Although we won't work with the commuting operator value in this work, we mention it here as it's an important value in the context of quantum operator assignments in the setting where nonlocality is represented by commuting subalgebras rather than the tensor product; see \cite{paddock2025satisfiability,culfmastel24} for more.

Several of the problems in this work can be formulated using the following template decision problem.

\begin{definition}
    Let $\mc{P}$ be a set and $\alpha,\beta:\mc{P}\to [0,1]$ be functions. For every $0<s\leq c\leq 1$, $(\alpha,\beta,\mc{P})_{c,s}$-Gapped-Decide is the decision problem with \emph{yes} instances $Y_{\alpha}=\{\phi\in \mc{P}:\alpha(\phi)\geq c\}$ and \emph{no} instances $N_{\alpha\beta}=\{\phi\in \mc{P}:\beta(\phi)<s \text{ and } \alpha(\phi)<c\}$. When $\alpha\leq \beta$ the promise $\alpha(\phi)<c$ in the``no'' case is redundant and we write $N_{\alpha\beta}=N_\beta$ to no longer depend on $\alpha$.
    Moreover, we say that $\mc{P}$ is $(\alpha,c)$-complete, and $(\beta,s)$-sound.
\end{definition}

We make several remarks. First off, the definitions appearing in \cite{mousavispirig24} regarding these gapped decision problems are incomplete, as the set of \emph{no} instances shrinks when $\alpha>\beta$. However, when $\alpha\leq \beta$, we have a number of so-called trivial reductions between the above problem instances.

\begin{lemma}\label{lem:triv_red}
Fix a set $\mc{P}$ and consider the pairs of functions $(\alpha_0,\beta_0)$ and $(\alpha_1,\beta_1)$ as defined above. For any $0<s<c\leq 1$, if $\alpha_0\leq \alpha_1\leq \beta_1\leq \beta_0$, then we have the following trivial reductions:
\begin{equation}
    \begin{tikzcd}
    \mc{P}_{\alpha_0,\beta_0} \arrow{r}{} \arrow{d}{} & \mc{P}_{\alpha_0,\beta_1} \arrow{d}{} \\
    \mc{P}_{\alpha_1,\beta_0}\arrow{r}{} & \mc{P}_{\alpha_1,\beta_1}
    \end{tikzcd}
\end{equation}
Moreover, if $\mc{P}$ is $(\alpha_0,c)$-complete then $\mc{P}$ is $(\alpha_1,c)$-complete. Similarly, if $\mc{P}$ is $(\beta_0,s)$-sound then $\mc{P}$ is $(\beta_1,s)$-sound.
\end{lemma}

The lemma follows from the following two propositions.

\begin{proposition}[Trivial $\alpha$-reductions]\label{prop:alpha_red}
If $\alpha_0\leq \alpha_1\leq \beta$ then there is a trivial reduction $\mc{P}_{\alpha_0,\beta}\to \mc{P}_{\alpha_1,\beta}$.
\end{proposition}

\begin{proof}
For every $0< s<c\leq 1$, if $\phi \in \mc{P}$ and $\beta(\phi)<s$, then $\alpha_i(\phi)<c$ for $i=1,2$, hence $N_{\alpha_0\beta}=N_{\alpha_1\beta}=N_\beta$. So the trivial reduction $\phi\in N_{\alpha_1\beta}$ implies $\phi \in N_{\alpha_0\beta}$, as desired. On the other hand, since $\alpha_0(\phi)\leq \alpha_1(\phi)$ for all $\phi \in \mc{P}$, we have that $Y_{\alpha_0} \subseteq Y_{\alpha_1}$, and therefore $\phi\in Y_{\alpha_0}$ implies $\phi\in Y_{\alpha_1}$, completing the proof.
\end{proof}

\begin{proposition}[Trivial $\beta$-reductions]\label{prop:beta_red}
    If $\beta_0\geq \beta_1$ there is a reduction $\mc{P}_{\alpha,\beta_0}\to \mc{P}_{\alpha,\beta_1}$.
\end{proposition}

\begin{proof}
     For every $0<s<c\leq 1$, the set of \emph{yes} instances $Y_\alpha$ is unchanged, so we only need to check the \emph{no} instances. Similar to the proof of \cref{prop:alpha_red}, the assumption $\beta_0\geq \beta_1$ implies that $N_{\alpha\beta_0}\supseteq N_{\alpha\beta_1}$ and therefore $\phi\in N_{\alpha\beta_1}$ implies $\phi \in N_{\alpha\beta_0}$.
\end{proof}

\begin{remark}\label{rem:triv_reductions}
    Note that there are always trivial $\beta$-reductions, even in the case that $\alpha_1>\beta$. However, we don't have trivial $\alpha$-reductions without the assumption that $\alpha_1\leq \beta$, as it could be possible that $\alpha_0\leq \beta$ and $\alpha_1\geq c$, in particular the intersection of $N_{\alpha_0\beta}$ and $Y_{\alpha_1}$ could be non-empty.
\end{remark}

\subsection{Constraint System Games and Weighted Algebras}

A \textbf{constraint system (CS)} over an alphabet $\Sigma$ consists of a set of variables $\mc{X}$ and a collection of constraints $\{\mc{C}_i\}_{i=1}^n$, where $\mc{C}_i=(\mc{U}_i,\mc{R}_i)$ for $1\leq i \leq n$, each $\mc{U}_i$ called the \textbf{context} is an ordered subset of $\mc{X}$, and the \textbf{relations} $\mc{R}_i$ are nonempty subsets of $\Sigma^{\mc{U}_i}$. When $\mc{U}_i=k$ we say that $\mc{C}$ is a $k$-ary constraint, and we write $k$-CS to denote a CS where each constraint has arity $k$. An \textbf{assignment} to $\mc{U}_i$ is an element of $\Sigma^{\mc{U}_i}$. Given a constraint $\mc{C}_i$ the \textbf{satisfying assignments} are the elements of $\mc{R}_i$, and assignments not in $\mc{R}_i$ are \textbf{unsatisfying assignments}. In the special case where the alphabet $\Sigma$ has size two, we say $S$ is a boolean constraint system (BCS). A \textbf{constraint language} is the family of constraint systems where each constraint comes from a restricted set of relations. An important example of a constraint language is $\SAT$, where each context contains three variables and each constraint is a conjunction of three literals. Such constraints are said to be in 3 conjunctive normal form (3CNF). Another example we mention here is the language of fixed degree $\SAT$, which we denote by $\SAT k$ for $k>0$, where each variable appears in exactly $k$ contexts.

A \textbf{constraint satisfaction problem} with completeness $c$ and soundness $s$, for $0<s<c\leq 1$, is a promise problem:
\begin{align*}
\tag{$k$-CSP$(m)_{c,s}$}
    \begin{minipage}{0.7\textwidth}
        Given a CS $S=(\mc{X},\{\mc{C}_i\}_{i=1}^n)$ where $|\mc{U}_i|=k$ for all $1\leq i \leq n$ and $|\Sigma|=m$:
        \begin{itemize}
            \item[(\emph{yes})] there exists a $\phi\in \Sigma^\mc{X}$ which satisfies at least $cn$ constraints,
            \item[(\emph{no})] each assignment $\phi\in \Sigma^\mc{X}$ satisfies less than $sn$ constraints.
        \end{itemize}
    \end{minipage}
\end{align*}

Similarly, a \textbf{succinct constraint satisfaction problem} with completeness $c$ and soundness $s$, for $0<s<c\leq 1$, is a promise problem:
\begin{align*}
    \tag{$\SuccinctCSP(m)_{c,s}$}
    \begin{minipage}{0.7\textwidth}
    Given a probabilistic Turing machine $M$ that samples the constraints of a CS $S=(X,\{(\mc{U}_i,\mc{R}_i)\}_{i=1}^n)$ according to some probability distribution $\pi:[n]\rightarrow[0,1]$:
    \begin{itemize}
        \item[(\emph{yes})] there is an assignment $f:X\rightarrow\Sigma$ such that $\Pr_{i\leftarrow\pi}[f|_{\mc{U}_i}\in \mc{R}_i]\geq c$,
        \item[(\emph{no})] for every assignment $\Pr_{i\leftarrow\pi}[f|_{\mc{U}_i}\in \mc{R}_i]<s$
    \end{itemize}
    \end{minipage}
\end{align*}

For convenience later, we write $\mathrm{SuccinctCSP}(\SAT)$ and $\mathrm{CSP}(\SAT)_{c,s}$ for the succinct (and non-succinct) constraint satisfaction problems where each CS is a $\SAT$ instance.  Additionally, it will be convenient to make use of the following graph associated with a CS, which exhibits which variables appear in which constraints in a constraint system.

\begin{definition}\label{def:connectivity-graph}
    Let $S=(\mc{X},\{\mc{C}_i\}_{i=1}^n)$ be a CS. The \textbf{connectivity graph} of $S$ is the bipartite graph $\Gamma(S)$ with left vertices $V_L(\Gamma(S))=\mc{X}$, right vertices $V_R(\Gamma(S))=[n]$, and edges $(x,i)\in E(\Gamma(S))$ if and only if $x\in\mc{U}_i$.
\end{definition}

Given a distribution $\mu:[n]\to \R_{\geq 0}$ and a CS $S$, the \textbf{constraint-variable CS game} is the nonlocal game $\mc{G}(S,\mu) =  ([n]\cup\mc{X},\{\mc{O}_i\}_i\in [n], \pi, V)$, where $O_i = \Sigma^{\mc{U}_i}$ when $i \in [n]$ and $O_i = \Sigma$ otherwise, $\pi(i,x) = \mu(i)/|\mc{U}_i|$ for $x\in \mc{U}_i$ and is 0 on other inputs, and $V(\phi,a|i,x) = 1$ if $\phi \in \mc{C}_i$ and $\phi(x) = a$ and is $0$ otherwise. Note that although the players have different question sets, the game can be symmetrized without changing the synchronous winning probability. In this game, one player receives $i\in[m]$ sampled from $\mu$ and the other receives a uniformly random variable $x\in \mc{U}_i$. To win, the first player must answer with a satisfying assignment $\phi\in \mc{C}_i$ and the second must answer $a\in\Sigma$ such that $a=\phi(x)$. When $S$ is boolean, we call that $\mc{G}(S,\pi)$ a BCS game.

The constraint-variable game is not naturally synchronous. To make it synchronous, the verifier must randomly choose which player to ask the constraint question and which player to ask the variable question, and also ask consistency check questions with some constant probability. This transformation preserves constant completeness-soundness gaps and guarantees that the players can win near-optimally with synchronous strategies due to Theorem 0.1 of \cite{marrakchi2023synchronous}, so we do not need to worry about it in practice.

Let $S$ be a $2$-CS over an alphabet $\Sigma$ and $\pi'$ be a probability distribution on $[m]$, define the \textbf{{$2$\nobreakdash-CS} game} $G_a(S,\pi')$ as the nonlocal game $(\mc{X},\Sigma,\nu_a,V_a)$, where $\nu_a(x,y)=\frac{\pi'(i)}{2}$ if $\mc{U}_i=\{x,y\}$ and $0$ otherwise, and $V_a(a,b|x,y)=1$ iff there exists $\phi\in \mc{R}_i$ such that $\phi(x)=a$ and $\phi(y)=b$. In this game, the referee samples $i$ from $\pi'$, and then asks each of the players one variable from the constraint. Each player responds with an assignment to the variable she received. They win if they have answered a satisfying assignment.

If the players are classical, then their strategy for either type of CS game can be taken to be deterministic. The responses of the player who received the variable question constitute an assignment to the constraint system. One can reason about reductions between classical CS games by working with these global assignments.

When the players have access to quantum resources, we can get a similar simplification when thinking about synchronous strategies using the weighted algebra formalism developed in \cite{MS24} and \cite{culfmastel24}. 

For a set of variables $\mc{U}$, let $\C\Z_k^{\ast\mc{U}}$ be the free algebra generated by order-$k$ unitaries labelled by the elements $x\in \mc{U}$, and let $\C\Z_k^{\mc{U}}$ be its abelianization. For any order-$k$ unitary $x$ and $a\in\Z_k$, write $\Pi_a^{(k)}(x)$ for the projector onto the $\omega_k^a$-eigenspace of $x$; where $k$ is clear, we suppress the superscript $(k)$. Given $\phi\in\Z_k^{\mc{U}}$, we define the element $\Phi_{\mc{U},\phi}^{(k)}\in\C\Z_k^{\ast \mc{U}}$ as
$$\Phi_{\mc{U},\phi}^{(k)}=\prod_{x\in V}\Pi_{\phi(x)}^{(k)}(x),$$
where as before the superscript~$(k)$ is suppressed where clear. We use the same notation for the image of $\Phi_{\mc{U},\phi}^{(k)}$ under any homomorphism when the homomorphism is clear. Given a constraint $(\mc{U},\mc{R})$ over the alphabet $\Z_k$, we write
$$\mc{A}(\mc{U},\mc{R})=\C\Z_k^{\mc{U}}/\gen{\Phi_{\mc{U},\phi}}{\phi\notin \mc{R}}.$$
The algebra $\mc{A}(\mc{U},\mc{R})$ is isomorphic to the $C^\ast$-algebra of functions on the finite set $\mc{R}$. Consequently, if $\rho: \mc{A}(\mc{U},\mc{R}) \to \mc{B}(\mc{H})$ is a $\ast$-representation, then $\{\rho(\Phi_{\mc{U},\phi})\}_{\phi\in \mc{R}}$ is a projective measurement on $\mc{H}$, and conversely if $\{M_\phi\}_{\phi\in \mc{R}}$ is a projective measurement on $\mc{H}$ then there is a $\ast$-representation $\rho: \mc{A}(\mc{U},\mc{R}) \to \mc{B}(\mc{H})$ with $\rho(\Phi_{\mc{U},\phi}) =M_{\phi}$.

Let $S = (\mc{X},(\mc{U}_i,\mc{R}_i)_{i=1}^n)$ be a CS. The \textbf{constraint-variable algebra} is the free product $\mc{A}_{c-v}(S)=\bigast_{i=1}^m\mc{A}(\mc{U}_i,\mc{R}_i)\ast\C\Z_k^{\ast \mc{X}}$. Let $\sigma_i:\mc{A}(\mc{U}_i,\mc{R}_i)\rightarrow\mc{A}_{c-v}(S)$ be the inclusion on the $i^{th}$ factor, and let $\sigma':\C\Z_k^{\ast \mc{X}}\rightarrow\mc{A}_{c-v}(S)$ be the inclusion on the final factor. To keep our formulas tidy, we'll often omit the $\sigma_i$ and $\sigma'$ when it's clear what subalgebra $\mc{A}(\mc{U}_i,\mc{R}_i)$ the element belongs to. By the GNS representation theorem, there is a correspondence between finite dimensional tracial states on the constraint variable algebra and synchronous quantum strategies for the constraint-variable game. A correlation $p\in C_q^s$ if and only if there is a finite dimensional tracial state $\tau$ on $\mc{A}_{c-v}(S)$ such that $p(\phi,a|i,x) = \tau(\Phi_{\mc{U}_i,\phi} \Pi_a(\sigma'(x)))$. A tracial state $\tau$ on $\mc{A}_{c-v}(S)$ is \textbf{perfect} if it corresponds to a perfect correlation. This happens if and only if $\tau(\Phi_{\mc{U}_i,\phi} \Pi_a(\sigma'(x))) = 0$ whenever $\phi(x) \neq a$. We are now ready to define the weight and defect, which gives us a way of tracking how close players' strategies are to being perfect at the level of tracial states on algebras.

\begin{definition}
    A \textbf{(finitely-supported) weight function} on a set $X$ is a function
    $\mu : X \to [0,+\infty)$ such that $\supp(\mu) := \mu^{-1}((0,+\infty))$
    is finite. A \textbf{weighted} $*$\textbf{-algebra} is a pair $(\mc{A},\mu)$
    where $\mc{A}$ is a $*$-algebra and $\mu$ is a weight function on $\mc{A}$. 

    If $\tau$ is a tracial state on $\mc{A}$, then the \textbf{defect of $\tau$} is
    \begin{equation*}
        \defect(\tau; \mu) := \sum_{a \in \mc{A}} \mu(a) \|a\|^2_{\tau},
    \end{equation*}
    where $\|a\|_{\tau} := \sqrt{\tau(a^* a)}$ is the $\tau$-norm.
    When the weight function is clear, we just write $\defect(\tau)$. 
\end{definition}
Since $\mu$ is finitely supported, the sum in the definition of the
defect is finite, and hence is well-defined. The support of the weight is chosen so that if $\defect(\tau;\mu) = 0$, then $\tau$ is perfect. 

For a probability distribution $\pi'$ on $[m]$, define the \textbf{weighted constraint-variable algebra} to be $\mc{A}_{c-v}(S,\pi')=(\mc{A}_{c-v}(S),\mu_{c-v,\pi'})$ where the weight function $\mu_{c-v}(\Phi_{V_i,\phi}(1-\Pi_{\phi(x)}(\sigma'(x))))=\frac{\pi'(i)}{|\mc{U}_i|}$ for all $x\in \mc{U}_i$, $\phi\in \mc{R}_i$, and $0$ on all other elements. Similarly, Given a probability distribution $\pi'$ on $[m]$, we define the \textbf{weighted assignment algebra} to be $\mc{A}_a(S)=\C\Z_k^{\ast \mc{X}}$, along with the weight function $\mu_{a,\pi'}(\Phi_{\mc{U}_i,\phi})=\pi'(i)$ for all $\phi\notin \mc{R}_i$, and $0$ on all other elements.

The point of these algebras is the following lemma.

\begin{lemma}\cite{culfmastel24}
    Let $S = (\mc{X},\{\mc{C}_i\}_{i=1}^m)$ be a CS, and let $\mu$ be a
    probability distribution on $[m]$. A tracial state $\tau$
    on $\mc{A}_{c-v}(S)$ is an $\epsilon$-perfect strategy for $\mc{G}(S,\pi)$ if and only
    if $\defect(\tau) \leq \epsilon$. If $S$ is a $2$-CS, then a tracial state $\tau_a$ on $\mc{A}_a(S)$ is an $\epsilon$-perfect strategy for $G_a(S,\pi)$ if and only
    if $\defect(\tau_a) \leq \epsilon$
\end{lemma}
See \cite{culfmastel24} for a more complete treatment of the weighted algebra formalism.

\section{The quantum smooth label cover game}

$\MIP$ protocols are a central example of a gapped decision problem. A \textbf{two-prover one-round $\MIP$ protocol} consists of a probabilistic Turing machine $Q$ and another Turing machine $V$, along with a family of nonlocal games
$G_x = (I_x,\{O_{xi}\}_{i \in I_x}, \pi_x, V_x)$ for $x \in \{0,1\}^*$,
such that 
\begin{itemize}
    \item on input $x$, the Turing machine $Q$ outputs $(i,j) \in I \times I$
        with probability $\pi_x(i,j)$, and  

    \item on input $(x,a,b,i,j)$, the Turing machine $V$ outputs $V_x(a,b|i,j)$. 
\end{itemize}

If the bit lengths of the questions and answers are bounded by functions $q(n)$ and $a(n)$ of $n = |x|$, respectively, and $\mc{P} = \{G_x\}_x$, we denote the decision problem $(\omega_q,\omega_q,\mc{P})_{c,s}$-Gapped-Decide by $\MIP^*(q(n),a(n))_{c,s}$. In general for $\MIP$ protocols, the completeness parameter $c$, and the soundness parameter $s$, may also depend on $n$ in an efficiently computable way. When we restrict the nonlocal games to be BCS games, we denote the decision problem by problem by BCS-$\MIP^*(q(n),a(n))_{c,s}$. When we further restrict the BCS games to be constraint-variable $\SAT$ games we denote the decision problem by $\mathrm{CSP}_{c-v}(\SAT)^*_{c,s}$ and we write $\SuccinctCSP_{c-v}(\SAT)^*_{c,s}$ for the succinctly presented version.

As mentioned in the previous section, the synchronous quantum value and the quantum value of synchronous nonlocal games are closely related. In our proofs of hardness, we deal with the synchronous quantum value as it is easier to reason about mathematically. This simplification is okay, since by Theorem 0.1 of \cite{marrakchi2023synchronous}, for constant $s$, the hardness of $(\omega^s_q,\omega^s_q,\mc{P})_{1,s}$-Gapped-Decide will be the same as that of $(\omega_q,\omega_q,\mc{P})_{1,s'}$-Gapped-Decide where $s'<0$ is also a constant.

Let $\mc{P}$ be an $\MIP$ protocol, and let $s>0$. The chain of inequalities in \Cref{eq:game-values} implies the trivial reductions in \Cref{fig:triv_reds}.

\begin{figure}[h!]
    \centering
    \begin{equation*}
    \begin{tikzcd}
    (\omega_c,\omega_{q}^s,\mc{P})_{1,s} \arrow{r}{}\arrow{d}{} & (\omega_c,\omega_{qo},\mc{P})_{1,s} \arrow{d}{}\arrow{r}{}&(\omega_c,\omega_c,\mc{P})_{1,s}\\
    (\omega_{qo},\omega_{q}^s,\mc{P})_{1,s} \arrow{r}{}\arrow{d}{} & (\omega_{qo},\omega_{qo},\mc{P})_{1,s}\arrow{r}{}& (\omega_{qo},\omega_{c},\mc{P})_{1,s}\\
    (\omega^s_{q},\omega_q^s,\mc{P})_{1,s} \arrow{r}{} & (\omega^s_{q},\omega_{qo},\mc{P})_{1,s}\arrow{r}{}& (\omega_{q}^s,\omega_c,\mc{P})_{1,s}
    \end{tikzcd}
\end{equation*}
    \caption{Trivial $\alpha$ and $\beta$ reductions between decision problems for the $\MIP$ protocol $\mc{P}$. The lack of certain $\alpha$-reductions is due to the technical issue noted in \cref{rem:triv_reductions}.}
    \label{fig:triv_reds}
\end{figure}

The main decision problems that we study center around the $2$-CS smooth label cover. Although we introduced the smooth label cover problem in the context of bipartite graphs, the problem can be generalized to any regular connected graph. That is, an instance of \textbf{Smooth Label Cover} is given by a tuple $(G, [n], [k], \Sigma)$ consisting of a (undirected) regular connected graph $G = (V, E)$, a label set~$[n]$ (for positive integer~$n$), and a set
\[
  \Sigma = \big((\pi_{ev}, \pi_{ew}) : e=(v,w)\in E\big)
\]
consisting of pairs of maps both from~$[n]$ to~$[k]$ associated with the endpoints of the edges in~$E$.
Given an \textbf{assignment} $A:V\to[n]$, we say that an edge~$e=(v,w)\in E$ is \textbf{satisfied} if $\pi_{ev}\big(A(v)\big) = \pi_{ew}\big(A(w)\big)$. The maps $\pi_{ev}$ are \textbf{smooth} in the sense that, for every vertex $v$ and every pair $a,a'\in[n]$ with $a\neq a'$, we have
\begin{equation*}
    \Pr_{w\sim v}[\pi_{ev}(a)\neq\pi_{ew}(a')]\approx 1.
\end{equation*}

Let $\SLC_{\alpha,\beta}(c,s)$ be the decision problem $(\omega_\alpha,\omega_\beta,\mc{P})_{c,s}$-Gapped-Decide, where $\mc{P}$ is an $\MIP^*$ protocol in which every instance is a $2$-CS game for a smooth label cover instance. When $\alpha = \beta$, we write $\SLC_{\alpha}(c,s)$.

The classical value of a Smooth Label Cover instance $\mathfrak{S}=(G, [n], [k], \Sigma)$ is the maximum fraction of satisfied edges over all possible assignments. A quantum assignment $S$ is given by a projection-valued measurement (PVM) $\{P_v^a\}_{a\in[n]}$ for each $v\in V$. We recall that a PVM consists of a collection of orthogonal projections $\{P^a\}_{a\in [n]}$ on a finite-dimensional Hilbert space that forms a resolution of identity $\sum_a P^a = I$. The value of such a quantum synchronous assignment $S$ is given by 
\begin{align*}
    \omega_q^s(\mathfrak{S},S)= \frac{1}{\abs{E}}\sum_{e=(v,w)\in E} \sum_{\substack{i,j\in [n]: \\ \pi_{ev}(i) = \pi_{ew}(j)}}\tr(P_v^i P_w^j).
\end{align*}
The quantum synchronous value of the Smooth Label Cover instance $\mathfrak{S}$, denoted $\omega_q^s(\mathfrak{S})$, is given by the supremum of $\omega_q^s(\mathfrak{S},S)$ over all possible synchronous strategies $S$ and finite-dimensional Hilbert spaces. An oracularizable quantum assignment $S$ consists of a set of PVMs, one for each vertex $v\in V$ of the graph $\{P_v^a\}_{a\in [n]}$which additionally satisfies the commutation relations $[P_v^a,P_w^b]=0$ for every edge $(v,w)\in E$ and labels $a,b\in[n]$. The value of such an oracularizable quantum assignment $S$ is given by
\begin{align*}
    \omega_{qo}(\mathfrak{S},S)= \frac{1}{\abs{E}}\sum_{e=(v,w)\in E} \sum_{\substack{i,j\in [n]: \\ \pi_{ev}(i) = \pi_{ew}(j)}}\tr(P_v^i P_w^j).
\end{align*}
The oracularizable quantum value of the Smooth Label Cover instance $\mathfrak{S}$, denoted $\omega_{qo}(\mathfrak{S})$, is given by the supremum of $\omega_{qo}(\mathfrak{S},S)$ over all possible oracularizable quantum assignments $S$ and finite-dimensional Hilbert spaces.
For each of these values, the value of a smooth label cover instance is equal to the value of its $2$-CS game, in this way $\mathrm{SLC}_{\alpha,\beta}(c,s)$ can be viewed as a problem of deciding the value of a family of nonlocal games or as the problem of deciding the value of instances of smooth label cover. To differentiate between instances and games, we write $\mc{G}_{\mathrm{SLC}(n)}$ to denote the $2$-CS nonlocal game associated with the smooth label cover instance $\mathfrak{S}$, despite that by our definition of values $\omega_t(\mc{G}_{\mathrm{SLC}(n)})=\omega_t(\mathfrak{S})$, for $t\in\{c,q^s,qo\}$.

\subsection{The Hardness of Quantum Smooth Label Cover}\label{sec:orac_smoothLC}

In this section, we establish the $\RE$-hardness of smooth label cover. Our reduction is based on the following hardness result involving the quantum value of a $\SAT5^*$ instance.

\begin{theorem}\label{thm:3sat-5}
    There exists a constant $0<s<1$ such that it is $\RE$-hard to distinguish whether the clause-variable game associated to a $\SAT5$ instance has quantum value 1 or at most~$s$.
\end{theorem}

We reserve the proof until \cref{sec:3SAT5}. Using \cref{thm:3sat-5}, we will prove the following theorem.

\begin{theorem}\label{thm:QSmoothLC}
    For any $0<s<1$, there exists large enough alphabets over which $\mathrm{SLC}_{q}(1,s)$ is $\RE$-hard.
\end{theorem}



The proof can be seen as a quantum analog of Theorem 3.5 in~\cite{guruswami-bypassing}. In particular, the proof of \cref{thm:QSmoothLC} is based on the construction in Section 2.2 of \cite{khot2}, and in Appendix A of \cite{guruswami-bypassing}. Before we give the proof, we require several intermediary results and definitions.


\begin{definition}\label{def:ord_JR_game}
    Let $B=(X,\{(V_i,C_i)\}_{i=1}^m)$ be a BCS. The \emph{$(J,R)$-dummy clause-variable game} is the nonlocal game $\mc{G}(B,J,R)$  with question set \begin{align*}
        I=[m]^{(J+1)R}\cup\set*{\vec{i}\in(X\cup[m])^{(J+1)R}}{\abs*{\set{q}{i_q\in X}}=R},
            \end{align*}
            answer set \begin{align*}
             O_{\vec{i}}=\bigtimes_{q=1}^{(J+1)R}\begin{cases}\{0,1\}&i_q\in X\\C_{i_q}&i_q\in[m]\end{cases},
            \end{align*}
    probability distribution $\pi(\vec{i},\vec{j})$ sampled by sampling $\vec{i}\in[m]^{(J+1)R}$ uniformly and then sampling $\vec{j}$ by sampling a subset $L\subseteq[(J+1)R]$ of size $R$ and $x_q\in V_{i_q}$ for all $q\in L$ uniformly and then letting $j_q=x_q$ if $q\in L$ and $j_q=i_q$ if $q\notin R$, and predicate $V(\vec{\sigma},\vec{\tau}|\vec{i},\vec{j})=1$ iff $\sigma_q=\tau_q$ if $i_q=j_q$ and $\sigma_q(j_q)=\tau_q$ else.

    Verifying the predicate function can be simplified by introducing the \emph{projection} $\pi^{\vec{j},\vec{i}}:O_{\vec{i}}\rightarrow O_{\vec{j}}$ defined as
    \begin{align*}
        \pi^{\vec{j},\vec{i}}(\vec{\sigma})_q=\begin{cases}\sigma_q&j_q=i_q\\\sigma_q(j_q)&j_q\in V_{i_q}\end{cases}.
    \end{align*}
    Then $V(\vec{\sigma},\vec{\tau}|\vec{i},\vec{j})=\delta_{\vec{\tau},\pi^{\vec{j},\vec{i}}(\vec{\sigma})}$.
\end{definition}

We have the following result about $\mc{G}(B,J,R)$.

\begin{lemma}\label{lem:dummy-game-value}
    If the synchronous quantum value of the constraint-variable game of a BCS $B$ is $<1$, there exists $\eta(R)\in\frac{1}{\exp(R)}$ such that the synchronous quantum value of $\mc{G}(B,J,R)$ is upper-bounded by $\eta(R)$. If the synchronous quantum value of the constraint-variable game of a BCS $B$ is $1$, then the synchronous quantum value of $\mc{G}(B,J,R)$ is $1$.
\end{lemma}

\begin{proof}
    Consider the modified version of $\mc{G}(B,J,R)$ where the predicate only checks that $\sigma_q(j_q)=\tau_q$ for those $q$ such that $j_q\in X$, and does not check for consistency on the remaining $q$. Then, it is easy to see that the synchronous quantum value of this game is an upper bound on the synchronous quantum value of $\mc{G}(B,J,R)$; and that this game is equivalent to the $R$-fold parallel repetition of the constraint-variable game of $B$. Hence, if the synchronous quantum value of a constraint-variable game of a BCS $B$ is less than $1$, then by the parallel repetition theorem for entangled projection games ~\cite{dinur2015parallel}, the synchronous quantum value of the parallel repetition is upper-bounded by $\frac{1}{\exp(R)}$. Hence, the same bound applies to $\mc{G}(B,J,R)$.

    Now, consider the case that the synchronous quantum value of the constraint-variable game of $B$ is $1$. Then there exists a perfect (quantum-approximate) strategy for this game consisting of Alice's PVMs $\{A_i^\sigma\}_{\sigma\in C_i}$ for $i\in[m]$ and Bob's PVMs $\{B_x^b\}_{b\in \{0,1\}}$, and a tracial state $\tr$. Then, construct the following strategy for $\mc{G}(B,J,R)$: let the PVMs $P_{\vec{i}}^{\vec{\sigma}}=P_{i_1}^{\sigma_1}\otimes\cdots\otimes P_{i_{(J+1)R}}^{\sigma_{(J+1)R}}$ where $P_{i_q}^{\sigma_q}=A_{i_q}^{\sigma_q}$ if $i_q\in[m]$ and $P_{i_q}^{\sigma_q}=B_{i_q}^{\sigma_q}$ if $j_q\in X$. With respect to the tracial state $\tr^{\otimes (J+1)R}$, it is clear that this is a perfect synchronous quantum approximate strategy, since it passes each of the consistency checks perfectly term-by-term.
\end{proof}

\begin{corollary}
    There exists an inverse exponential $\eta(R)$ such that it is $\RE$-hard to decide whether $\mc{G}(B,J,R)$ for a $\SAT5$ instance $B$ has synchronous quantum value $1$ or $\leq\eta(R)$.
\end{corollary}

\begin{proof}
    By \cref{thm:3sat-5}, there exists $s<1$ such that it is $\RE$-hard to decide if the synchronous quantum value of the constraint-variable game of a $\SAT5$ instance $B$ is $1$ or $<s$. Using \cref{lem:dummy-game-value} gives the wanted reduction.
\end{proof}

\begin{definition}\label{def:connectivity-graph-game}
    Let $\mc{G}=(I,\{O_i\},\mu,V)$ be a nonlocal game. The \emph{connectivity graph of $\mc{G}$} is the bipartite graph $\Gamma(\mc{G})$ with left vertices $V_L(\Gamma(\mc{G}))=\set*{i\in I}{\sum_j\mu(i,j)>0}$, right vertices $V_R(\Gamma(\mc{G}))=\set*{j\in I}{\sum_i\mu(i,j)>0}$, and edges $(i,j)\in E(\Gamma(\mc{G}))$ if and only if $\pi(i,j)>0$.
\end{definition}

In particular, the connectivity graph of the constraint-variable game $\mc{G}(S,\mbbm{u}_n)$ is equal to the connectivity graph of the CS $S$ (see \cref{def:connectivity-graph}). Hence, the connectivity graph $\Gamma(\mc{G}(B,J,R))$ for $B$ a $\SAT5$ instance has $n=7^{(J+1)R}$ left vertices and $k=2^R7^{JR}$ right vertices. In the following, we often equate the left vertices of $\Gamma(\mc{G}(B,J,R))$ with $[n]$ and the right vertices with $[k]$.





The reason we introduced the ``dummy'' clauses is the following lemma from \cite{khot2}.



\begin{lemma}[\cite{khot2}, Lemma 2.4]\label{lem:khot-smoothness-2}
    For fixed $\vec{i}\in[m]^{(J+1)R}$ and any two distinct assignments $\vec{\sigma}_1$ and $\vec{\sigma}_2$ to $\vec{i}$,
    \begin{align*}
        \Pr_{\vec{j}}\squ[\big]{\pi^{\vec{j},\vec{i}}(\vec{\sigma}_1)\neq\pi^{\vec{j},\vec{i}}(\vec{\sigma}_2)}\geq 1-\frac{1}{J},
    \end{align*}
    where $\vec{j}$ is distributed according to the conditional distribution of $\mu$ for fixed $\vec{i}$, \textit{i.e.} $\Pr[\vec{j}=\vec{j}_0]=\mu(\vec{i},\vec{j}_0)/\sum_{\vec{k}}\pi(\vec{i},\vec{k})$.
\end{lemma}

\begin{proof}
    By assumption, there exists some $q_0\in[(J+1)R]$ such that $(\sigma_1)_{i_{q_0}}\neq(\sigma_2)_{i_{q_0}}$. As such, if $j_{q_0}=i_{q_0}$, then $\pi^{\vec{j},\vec{i}}(\vec{\sigma}_1)\neq\pi^{\vec{j},\vec{i}}(\vec{\sigma}_2)$, since the projection preserves the $q_0$-th term. Let $L=\set*{q}{j_q\neq i_q}\subseteq[(J+1)R]$ be a set-valued random variable. Then, we have the lower bound $\Pr_{\vec{j}}\squ[\big]{\pi^{\vec{j},\vec{i}}(\vec{\sigma}_1)\neq\pi^{\vec{j},\vec{i}}(\vec{\sigma}_2)}\geq\Pr[q_0\notin L]$. The number of subsets of size $R$ in $[(J+1)R]$ is $\binom{(J+1)R}{R}$ and the number that do no include $q_0$ is $\binom{(J+1)R-1}{R}$, so the probability
    \begin{align*}
        \Pr[q_0\notin L]&=\frac{\binom{(J+1)R-1}{R}}{\binom{(J+1)R}{R}}=\frac{((J+1)R-1)!(JR)!R!}{(JR-1)!R!((J+1)R)!}\\
        &=\frac{JR}{(J+1)R}=1-\frac{1}{J+1}\geq 1-\frac{1}{J}.\qedhere
    \end{align*}
\end{proof}

We state the following lemma for later use. 



\begin{proposition}\label{prop:preimage-projectionmaps-2}
    For any $\vec{i}\in[m]^{(J+1)R}$ and $\vec{j}$ such that $\mu(\vec{i},\vec{j})>0$, we have for any assignment $\vec{\tau}$ to $\vec{j}$ that $\abs{(\pi^{\vec{j},\vec{i}})^{-1}(\tau)}\leq 4^R$.
\end{proposition}

\begin{proof}
    For any $x\in\{0,1\}$, there are at most $4$ ways to extend it to a satisfying $\SAT$ clause. Hence, since there are $R$ indices of $\tau$ to be extended to a satisfying assignment to $\vec{i}$, there are at most $4^R$ ways to do it.
\end{proof}

We are now ready to give the proof of Theorem \ref{thm:QSmoothLC}.

\begin{proof}[Proof of \cref{thm:QSmoothLC}]


We now want to construct a Smooth Label Cover instance $(G,[n],[k],\Sigma)$ following the construction in Appendix A of \cite{guruswami-bypassing}. As above, consider a bipartite graph $\Gamma(\mc{G}(B,J,R))$. Associated to each edge $(\vec{i},\vec{j})$ there is a map $\pi^{\vec{j},\vec{i}}$ which satisfies the smoothness property (Lemma \ref{lem:khot-smoothness-2}). Moreover, the graph $\Gamma(\mc{G}(B,J,R))$ is bi-regular as the $\SAT5$ instances we started with have the property that every variable appears exactly in 5 clauses, giving that the left degree is $3^R$ and the right degree is $5^R$.

Now, let $V$ be the set of left vertices of $\Gamma(\mc{G}(B,J,R))$, and connect $\vec{i},\vec{i}'\in V$ by an edge if there exists a right vertex $\vec{j}$ that is adjacent to both of them in $\Gamma(\mc{G}(B,J,R))$. For an edge $e=\{\vec{i},\vec{i'}\}$ corresponding to a right vertex $\vec{j}$, define the projection maps $\pi_{e,\vec{i}}=\pi_{\vec{j},\vec{i}}$ and $\pi_{e,\vec{i}'}=\pi_{\vec{j},\vec{i}'}$.

We claim that the constructed smooth label cover instances $(G=(V,E),[n],[k],\Sigma)$ satisfy the following four structural properties:
	\begin{enumerate} 

	\item (Smoothness)\ \ %
For every vertex $v \in V$ and distinct~$i,j\in [n]$, we have 
\begin{align}\label{eq:smoothness-2}
    \Pr_{w\sim v} \left[\pi_{ev}(i) = \pi_{ew}(j)\right] \leq \frac{1}{J}\,.
\end{align}

	\item For every vertex $v\in V$, edge $e\in E$ incident on $v$, and $i \in [k]$, we have $|\pi_{ev}^{-1}(i)|\leq t:=4^R $;
that is, at most $t$ elements in $[n]$ are mapped to
the same element in $[k]$.

    \item The degree of the regular graph $G$ is a constant depending only on $R$ and $J$.

	\item (Weak Expansion)\ \ %
For any $\delta > 0$ and vertex subset $V'\subseteq V$ such that
$|V'| = \delta\cdot |V|$, the number of edges between the vertices in $V'$
 is at least $\delta^2 |E|$.  
	\end{enumerate} 

The smoothness property follows immediately from \cref{lem:khot-smoothness-2}. The second property follows immediately from \cref{prop:preimage-projectionmaps-2}. The third property is also immediate. Indeed, the regularity of the graph $G$ follows directly from the bi-regularity of $\Gamma(\mc{G}(B,J,R))$. The weak expansion property is a standard notion (see \cite{guruswami-bypassing}), and we give a proof of it in~\cref{sec:weak-expansion}.

We move on to show that it is $\RE$-hard to distinguish between the following two cases: $\omega_{q}(\mathfrak{S})=1$, or $\omega_q^s(\mathfrak{S})\leq 1/\exp(R)$. For completeness, let $\mc{S}$ denote a quantum strategy winning perfectly at the synchronous version of the clause variable $(J,R)$-dummy game associated to $B$ as in \cref{sec:nonlocalgames}. Since the game is synchronous, there must exist a perfect quantum synchronous assignment $\mc{S}_s$. Let us denote it $\mc{S}_s=(P_u, Q_v)$, $P_u=\{P_u^l\}_{l\in [k]}$ is a PVM for every $u\in U$ and $Q_v=\{Q_v^i\}_{i\in [n]}$ is a PVM for every $v\in V$. Then we have
    \begin{align}\label{eq:assumption-completeness}
        1=\omega_{q}(\mc{G}(B,J,R),\mc{S}_s)=\expect_{(u,v)\in E} \sum_{i\in [n]} \tr\Paren{P_u^{\pi_{v,u}(i)} Q_v^i}
    \end{align}

    Now consider the quantum assignment for $\mathfrak{S}$ given by the operators $\mc{S}'=\{Q_v: v\in V\}$. We compute its value:
    \begin{align*}
        \omega_{q}(\mathfrak{S},\mc{S}')
        &= \expect_{u\in U, v_1,v_2\in N_u} \sum_{\substack{i,j\in [n]: \\ \pi_{v_1,u}(i) = \pi_{v_2,u}(j)}}\tr\Paren{Q_{v_1}^i Q_{v_2}^j}\\
        &= \expect_{u\in U, v_1,v_2\in N_u} \sum_{l\in [k]} \tr\Paren{\sum_{i \in \pi_{v_1,u}^{-1}(l)}Q_{v_1}^i \sum_{j \in \pi_{v_2,u}^{-1}(l)}Q_{v_2}^j}\\
        &=\expect_{u\in U, v_1,v_2\in N_u} \sum_{l\in [k]} \tr\Paren{\Tilde{Q}_{v_1}^l \Tilde{Q}_{v_2}^l},
    \end{align*}
    where we defined $\Tilde{Q}_{v_1}^l=\sum_{i \in \pi_{v_1,u}^{-1}(l)}Q_{v_1}^i$ and similarly for $\Tilde{Q}_{v_2}^l$.

    Now we note that $\{\Tilde{Q}_{v_1}^l\}_{l\in [k]}$ defines a PVM. Indeed we have $\sum_{l\in [k]} \Tilde{Q}_{v_1}^l = \sum_l \sum_{i \in \pi_{v_1,u}^{-1}(l)}Q_{v_1}^i =1$. Moreover, since $\{Q_v^i\}_{i\in [n]}$ forms an orthonormal set of projections, it is clear that $\{\Tilde{Q}_{v_1}^l\}_{l\in [k]}$ will as well form an orthonormal set of projections. Therefore, we can use Lemma 34 of~\cite{mousavispirig24}:
    \begin{align*}
        \omega_{q}(\mathfrak{S},S)&=\expect_{u\in U, v_1,v_2\in N_u} \sum_{l\in [k]} \tr\Paren{\Tilde{Q}_{v_1}^l \Tilde{Q}_{v_2}^l}\\
        &\geq 2 \expect_{u\in U, v_1,v_2\in N_u} \sum_{l\in [k]} \tr\Paren{\Tilde{Q}_{v_1}^l P_u^l + \Tilde{Q}_{v_2}^l P_u^l} - 3\\
        &= 1,
    \end{align*}
    where $N_u$ is the set of vertices adjacent to $u$, and we used that 
    \begin{align*}
        \expect_{u\in U, v_1,v_2\in N_u} \sum_{l\in [k]} \tr\Paren{\Tilde{Q}_{v_1}^l P_u^l} &= \expect_{(u,v_1)\in E} \sum_{l\in [k]} \sum_{i \in \pi_{v_1,u}^{-1}(l)} \tr\Paren{Q_{v_1}^i P_u^l}\\
        &= \expect_{(u,v_1)\in E} \sum_{i\in [n]} \tr\Paren{Q_{v_1}^i P_u^{\pi_{v_1,u}(i)}},
    \end{align*}
    along with the assumption from equation \eqref{eq:assumption-completeness}. This completes the proof of completeness.

For soundness, let $\mc{L}=(\{Q_v^i\}_{i\in [n]})$ be a quantum synchronous assignment for $\mathfrak{S}$ which satisfies:
    \begin{align}\label{eq:assumption-soundness}
        \omega_{q}(\mathfrak{S},\mc{L})= \expect_{u\in U, v_1,v_2\in N_u} \sum_{\substack{i,j\in [n]: \\ \pi_{v_1,u}(i) = \pi_{v_2,u}(j)}}\tr\Paren{Q_{v_1}^i Q_{v_2}^j} >s,
    \end{align}
    for some $s<1$.
    Define the operators $P_u^l = \expect_{v \in N_u} \sum_{i \in \pi_{e,v}^{-1}(l)} Q_v^i$. These operators form a POVM as they are clearly positive and satisfy
    \begin{align*}
        \sum_{l\in [k]} P_u^l =  \expect_{v \in N_u} \sum_l \sum_{i \in \pi_{e,v}^{-1}(l)} Q_v^i = 1,
    \end{align*}
    as $\{Q_v^i\}_i$ is a PVM. 
    Using Naimark's dilation theorem, we can replace $P_u$ with a PVM that performs at least as well.
Now, consider the assignment $\mc{S}=(P_u, Q_v)$ for $\mc{G}(B,J,R)$. This assignment has value:
    \begin{align*}
        \omega_q(\mc{G}(B,J,R),\mc{S})&=\expect_{(u,v_1)\in E} \sum_{i\in [n]} \tr\Paren{P_u^{\pi_{v_1,u}(i)} Q_{v_1}^i}\\
        &= \expect_{(u,v_1)\in E} \expect_{v_2 \in N_u} \sum_{i\in [n]} \sum_{j \in \pi_{e,v_2}^{-1}(\pi_{v_1,u}(i))} \tr\Paren{  Q_{v_2}^j Q_{v_1}^i}\\
        &= \expect_{u\in U, v_1,v_2\in N_u} \sum_{\substack{i,j\in [n]: \\ \pi_{v_1,u}(i) = \pi_{v_2,u}(j)}}\tr\Paren{Q_{v_1}^i Q_{v_2}^j}>s,
    \end{align*}
    using the assumption from equation \eqref{eq:assumption-soundness}. By \cref{lem:dummy-game-value} we conclude that the $s\leq 1/\exp(R)$ and the claim follows.
\end{proof}

\begin{remark}
    In the argument above, we note that properties (2)-(4) are not necessary to establish the proof of the theorem statement. However, in establishing these results, it does mean that in addition to $\RE$-hardness, the instances of smooth label cover coming from this reduction also have these additional properties, which made them amenable to the specific applications in \cite{guruswami-bypassing}. 
\end{remark}

\subsection{The Hardness of Quantum Oracularized Smooth Label Cover}

In this section, we show an analog of Theorem \ref{thm:QSmoothLC} where we consider the quantum oracularized value instead of the quantum synchronous value.

\begin{corollary}\label{thm:QOSmoothLC}
    For any $0<s<1$, there exists large enough alphabets over which $\mathrm{SLC}_{qo}(1,s)$ is $\RE$-hard.
\end{corollary}


This result follows directly from the following theorem and the trivial reduction $\mathrm{SLC}_{qo,q}(1,s)$, see \cref{fig:triv_reds}.

\begin{theorem}\label{thm:QQOSmoothLC}
    For any $0<s<1$, there exists large enough alphabets over which $\mathrm{SLC}_{qo,q}(1,s)$ is $\RE$-hard.
\end{theorem}


To prove this theorem, we start from a slightly altered version of Theorem \ref{thm:3sat-5}. In particular, the following lemma shows that we can alter the $\SAT5$ game so that the resulting optimal strategies for the SLC instance obtained by the reduction outlined in the proof of \cref{thm:QSmoothLC} will be oracularizable.


\begin{lemma}\label{lem:3sat5_commutation}
    There exists a constant $0<s<1$ such that it is $\RE$-hard to distinguish whether the clause-variable game associated with a $\SAT-10$ instance has synchronous quantum value 1 or at most~$s$. Moreover, when the instance has value 1, then there exist binary observables $\{A_x\}_{x\in X}$ corresponding to each variable $x\in X$ in the instance, which satisfy the following properties:

    \begin{enumerate}
        \item the observables satisfy the constraints
        \item $[A_x,A_y]=0$ if there exists $i\in[m]$ such that $x,y\in V_i$,
        \item $[A_x,A_y]=0$ if there exist $i,j\in[m]$ and $z\in X$ such that $x,z\in V_i$ and $y,z\in V_j$.
    \end{enumerate}
\end{lemma}


\begin{proof}
The only non-trivial part of the statement is the third property. Let $S$ be a $\SAT5$ instance and $\Gamma=\Gamma(S)$ be its constraint graph with constraint (left) vertices $u\in U$ and (right) variables vertices $x\in X$. Starting from $\Gamma$ consider the following graph $\Gamma'$ which is constructed as follows: Let $N(u)\subset X$ be the neighbourhood of $u$. These represent the variables in the corresponding constraint. For each $u\in U$, add a ``local'' copy $x_u$ to $X$ for each ``global'' vertex $x\in N(u)$ (i.e.~5 copies for each variable vertex, as $S$ is a $\SAT5$ instance). For each copy of $N(u)$ added to $X$, we remove $u$ and add 7 copies of the vertex $u$ to $U$ and, for each new copy $u'$, an edge from $u'$ to each vertex $x_u$ in the copy of $N(u)$. Now, for each added variable vertex $x_u$, add 2 vertices $w_{x_u}$ and $v_{x_u}$ to the constraint vertices $U$, and edges connecting both $x_u$ and $x$ to $w_{x_u}$ and $v_{x_u}$. We now observe that each vertex in $X$ has degree 10. For convenience later, we let $U'\subset U$ denote the vertices in $\Gamma'$ incident to only ``local'' variable vertices. Likewise, we let $V\subset U$ and $W\subset U$ denote the subsets of vertices $w_{x_u}$ and $v_{x_u}$, respectively. Notably, $U=U'\sqcup V \sqcup W$ are the constraint vertices of $\Gamma'$.

Now, consider the CS $S$ and the connectivity graph $\Gamma'$ built from $\Gamma$ as above. From $S$ and $\Gamma'$ we construct a new CS $S'$ as follows. For each $u'\in U'$, there is a $\SAT$ clause $u$ from $S$ that was copied to create $u$. Add to $S'$ a copy of that constraint with $x_u$ taking the role of $x$ for each $x \in N(u)$. Next, for the vertices $v_{x_u}\in V$, add the constraint $x\lor x\lor \lnot x_{u}$ to $S'$. Finally, for the vertices $w_{x_u}\in W$, add the constraint $x_{u}\lor x_{u}\lor \lnot x$ to $S'$.

We now observe that for the constructed CS $S'$ we have that $\Gamma(S')=\Gamma'$. Next, suppose there exists a quantum perfect strategy for the constraint-variable game corresponding to $S$. Denote $A_x$ for the observable corresponding to $x\in X$ in this perfect strategy. Now, extend this quantum assignment to the variables of $S'$ by taking the observables for $x_u$ to be $A_{x_u}=A_x$. By construction, this satisfies both the 3CNF constraints on the labelled variables and the new equality constraints. Hence, it gives a quantum satisfying assignment to $S'$ --- hence the reduction is complete. Moreover, we note that this assignment also satisfies conditions (2) and (3) of the statement. First, if $x$ and $y$ belong to the same context of $S'$, then either $x=a_u$ and $y=b_u$ for some $a,b$ belonging to the same context of $S$ and hence $[A_x,A_y]=[A_a,A_b]=0$; or $y=x_u$ (or vice-versa) so $A_x=A_y$ and hence they commute. Second, suppose $x$ and $y$ are variables of $S'$ that each share a constraint with a variable $z$. By construction, either the constraint containing $x$ and $z$ or the constraint containing $y$ and $z$ must contain a ``global'' variable $w$. Without loss of generality, let us assume we are in the first case. That implies that $x=w$ or $x=w_u$, and $z=w$ or $z=w_u$, so $A_x=A_w=A_z$. As $z$ is in the same constraint as $y$, this means $[A_x,A_y]=[A_z,A_y]=0$.

For soundness, one can observe that the transformation between $S$ to $S'$ is an instance of subdivision, and therefore we can appeal to \cite[Theorem 5.2]{culfmastel24}. In particular, this implies that if the constraint-variables CS games corresponding to $S$ has quantum synchronous value less than $s$, then the constraint variable game corresponding to $S'$ has quantum synchronous value less than $C_0s$, for some universal constant $C_0>1$.

\end{proof}

Now we can establish \cref{thm:QQOSmoothLC}. Essentially, the idea is that the $2$-oracularizability condition in \cref{lem:3sat5_commutation} will ensure that the resulting quantum synchronous strategy for the smooth label cover instance will be oracularizable.

\begin{proof}[Proof of \cref{thm:QQOSmoothLC}]
    This proof is similar to the proof of \cref{thm:QSmoothLC}. The only difference is we start the reduction from the instances in \cref{lem:3sat5_commutation} rather than \cref{thm:3sat-5}. The proof of soundness is exactly the same. We only need to show that in the completeness argument that the assignment we construct for the smooth label cover instance is a quantum oracularizable assignment (meaning that it satisfies the desired commutation relations). 


    Consider an instance $B$ of $\SAT-10$ from \Cref{lem:3sat5_commutation} which has quantum value 1. Let $A_x$ denote the operators as in the statement of \Cref{lem:3sat5_commutation}, and let $\{A_x^b\}_{b\in\{0,1\}}$ denote the corresponding PVMs. Then, by the construction in the proof of \cref{lem:dummy-game-value}, Alice's operators in the corresponding perfect strategy for $\mc{G}(B,J,R)$ are
    $$P_{\vec{i}}^{\vec{\sigma}}=\bigotimes_{q=1}^{(J+1)R}\prod_{x\in V_{i_q}}A_x^{\sigma_q(x)}.$$
    We have that, for questions $\vec{i},\vec{i}'$ of $\mc{G}(B,J,R)$ such that there exists $\vec{j}$ with $\pi(\vec{i},\vec{j}),\pi(\vec{i}',\vec{j})>0$, $[P_{\vec{i}}^{\vec{\sigma}},P_{\vec{i}'}^{\vec{\sigma}'}]=0$. This is because, for all $q\in[(J+1)R]$, either $j_q=i_q=i_q'$ or $j_q\in V_{i_q}\cap V_{i_q'}$, and therefore for each $x\in V_{i_q}$ and $y\in V_{i_q'}$ either $x$ and $y$ are in the same context, or there exists $z$ that shares a context with both $x$ and $y$. Hence, by properties (2) and (3) of the strategy induced by the $A_x$ from \cref{lem:3sat5_commutation}, $[A_x^{\sigma_q(x)},A_y^{\sigma_q'(y)}]=0$. Then, using the construction of \cref{thm:QSmoothLC}, the strategy $\mc{S}'$ for $\mathfrak{S}$ is given by $Q_v^i=P_{\vec{i}}^{\vec{\sigma}}$, where $v\in V$ corresponding to $\vec{i}$ and $i\in[n]$ corresponding to $\vec{\sigma}$. It follows from the above commutation relation on the $P_{\vec{i}}^{\vec{\sigma}}$ that this strategy is oracularizable. This is a perfect strategy for $\mathfrak{S}$ by the same argument as in the proof of \cref{thm:QSmoothLC}. 
    
\end{proof}

\section{RE-hardness of 3SAT* with fixed degree}\label{sec:3SAT5}

In this section, we give a proof of \cref{thm:3sat-5} inspired by the classical approach in \cite{Feige98threshold}. Our starting point is the following result.

\begin{theorem}[\cite{culfmastel24}]\label{thm:3sat-cm24}
    There exists a constant $0<s<1$ such that there is a polynomial-time reduction from the halting problem to $\mathrm{SuccinctCSP}_{c-v}(\mathrm{3SAT})_{1,s}^\ast$.
\end{theorem}

First, we want to show hardness of $\SAT$ constraint-variable $\BCSMIP^\ast$ protocols with uniform question distributions. However, due to a later reduction in this section (\cref{thm:replacement-soundness}), we cannot work with the succinct presentation of nonlocal games, but rather the direct presentation. To do this, we move the exponential-time reduction to the level of RE, where the computation time is immaterial, by making use of a variant of the halting problem.

\begin{definition}
    Fix a universal Turing machine $M$. The \emph{halting problem} is the language $\mathrm{HALTING}$ of all strings $x\in\{0,1\}^\ast$ such that $M$ halts on input $x$. For a function $f:\N\rightarrow\N$ such that $f(n)\geq n$, the \emph{$f$-padded halting problem} is the language $\mathrm{PADDEDHALTING}_f=\set{x0^{f(|x|)-|x|}}{x\in\mathrm{HALTING}}$. 
\end{definition}

There is a $O(f)$-time reduction from the halting problem to the padded halting problem by padding with $0$, and a polynomial-time reduction from the padded-halting problem to the halting problem by truncating. Therefore $\mathrm{PADDEDHALTING}_f$ is $RE$-complete.

Now, we show that constraint-variable BCS nonlocal games with uniform distribution are $\RE$-complete in the following way.

\begin{lemma}\label{lem:padded-to-3sat}
    There exists $s\in(0,1)$ such that there is a polynomial-time reduction from the $\exp$-padded halting problem to $\mathrm{CSP}_{c-v}(\mathrm{3SAT})^\ast_{1,s}$.
\end{lemma}

In fact, this step can still be shown at the level of succinctly-presented games, \textit{i.e.} there is a polynomial-time reduction from the halting problem to $\BCSMIP^\ast_{unif}(\polylog(n),O(1))_{1,s}$. We show this in \cref{sec:uniformmarginals}.

\begin{proof}
    Due to \cref{thm:3sat-cm24}, there exists a polynomially-bounded function $p$ such that the reduction to succinct entangled 3SAT maps from a string $x$ --- seen as input to a universal Turing machine $M$ --- to Turing machines that sample a probability distribution on constraints from 3SAT instance $S_x=(X,\{(V_i,C_i)\}_{i=1}^{2^{p(|x|)}}))$. Now consider the reduction from the $2^p$-padded halting problem defined as follows. For any string $y$, truncate it to the first $p^{-1}(\log|y|)$ bits and call that substring $x$. If $y\neq x0^{p(|x|)-|x|}$, then map $y$ to a fixed unsatisfiable constraint system $S_y'=S_0$. Else, since $x$ is an instance of the halting problem, the reduction of \cref{thm:3sat-cm24} maps it to a probability distribution $\pi$ on constraints of a 3SAT instance $S_x$ with $2^{p(|x|)}=|y|$ constraints. Now, by sampling from $\pi$ $\poly(|y|)$ times, we learn the probability distribution up to negligibly small error in $|y|$. Now, let $S_y'=(X,\{(V_{i,j},C_i)\}_{i,j}))$ be the $3$SAT instance defined from $S_x$ by repeating each constraint $i$ $\ceil{|y|^2\pi(i)}$ times. $S'_y$ has $m_y'=\sum_i\floor{|y|^2\pi(i)}$ constraints, which is bounded as $m_y'\leq\sum_i(|y|^2\pi(i)+1)=|y|^2+|y|$ and $m_y'\geq\sum_i|y|^2\pi(i)=|y|^2$.

    We claim that the map $y\mapsto S_y'$ is a polynomial-time reduction from the $2^p$-padded halting problem to $\mathrm{CSP}_{c-v}(\mathrm{3SAT})_{1,(1+s)/2}^\ast$. It is clear that the reduction is polynomial-time in $|y|$. If $y$ is a yes instance of the $2^p$-padded halting problem, then $x$ is a yes instance of the halting problem, so $S_x$ has a perfect quantum satisfying assignment. Since $S_y'$ has the same constraints as $S_x$, it is also perfectly quantum satisfiable. Now, if $y$ is a no instance of the halting problem, either $y\neq x0^{2^{p(|x|)}-|x|}$ or $x$ is a no instance of the halting problem. In the former case, $S_y'=S_0$ can be chosen to have large enough defect. In the latter case, $\defect(\tau)\geq 1-s$ for any finite-dimensional tracial state on $\mc{A}_{c-v}(S_x,\pi)$. Now, let $\tau'$ be a finite-dimensional tracial state on $\mc{A}_{c-v}(S_y',\mathbbm{u}_{m_y'})$, and let $\varphi':\mc{A}_{c-v}(S_y')\rightarrow \mc{M}$ be the GNS representation, with tracial state $\rho$. Now, let $\varphi$ be the representation of $\mc{A}_{c-v}(S_x)$ defined as $\varphi(\Pi_{b}(\sigma'(a)))=\varphi'(\Pi_{b}(\sigma'(a)))$ and $\{\varphi(\Phi_{V_{i},\phi})\}_{\phi\in C_i}$ as the PVM in $\mc{M}$ that minimises $\sum_{\phi\in C_i}\sum_{a\in V_i}\rho(\varphi(\Phi_{V_{i},\phi})(1-\varphi(\Pi_{\phi(a)}(\sigma'(a)))))$; let $\tau=\rho\circ\varphi$. $\tau$ is a tracial state on $\mc{A}_{c-v}(S_x,\pi)$, and by extremality of the PVMs within the POVMs, its defect can be upper bounded as
    \begin{align*}
        \defect(\tau)&=\sum_{i=1}^{|y|}\frac{\pi(i)}{|V_i|}\sum_{\varphi\in C_i}\sum_{a\in V_i}\rho(\varphi(\Phi_{V_{i},\phi})(1-\varphi(\Pi_{\phi(a)}(\sigma'(a)))))\\
        &\leq\sum_{i=1}^{|y|}\frac{\pi(i)}{|V_i|}\sum_{\varphi\in C_i}\sum_{a\in V_i}\rho\parens[\Big]{\frac{1}{\ceil{|y|^2\pi(i)}}\sum_{j=1}^{\ceil{|y|^2\pi(i)}}\varphi'(\Phi_{V_{i,j},\phi})(1-\varphi'(\Pi_{\phi(a)}(\sigma'(a))))}\\
        &\leq\sum_{i,j}\frac{\pi(i)}{\ceil{|y|^2\pi(i)}}\frac{1}{|V_{i,j}|}\sum_{\varphi\in C_i}\sum_{a\in V_i}\tau'(\Phi_{V_{i,j},\phi}(1-\Pi_{\phi(a)}(\sigma'(a)))).
    \end{align*}
    Now, $\ceil{|y|^2\pi(i)}\geq |y|^2\pi(i)$, so $$\frac{\pi(i)}{\ceil{|y|^2\pi(i)}}\leq\frac{1}{|y|^2}=\frac{m_y'}{|y|^2}\mathbbm{u}_{m_y'}(i,j)\leq\parens*{1+\frac{1}{|y|^2}}\mathbbm{u}_{m_y'}(i,j).$$
    Hence, $\defect(\tau)\leq\parens*{1+\frac{1}{|y|^2}}\defect(\tau')\leq 2\defect(\tau')$. Hence, for every $\tau'$, $\defect(\tau')\geq\frac{1-s}{2}$, giving the wanted gap.
\end{proof}

Now, we show that the $\BCSMIP^\ast$ protocol can be reduced to an instance of entangled $\SAT5$, which denotes a $\SAT$ CSP where every variable appears exactly in $5$ clauses, \emph{i.e.} the degree of every left vertex in its connectivity graph (see \cref{def:connectivity-graph}) is $5$. In other words, we show that we can fix the degree of $\SAT^*$ while preserving a constant soundness gap. Although our approach is inspired by the classical proof of \cite{Feige98threshold}, establishing the proof in the quantum case requires some additional technical results, in particular a novel use of expander graphs, which is not needed to establish the classical analogue.


\begin{definition}
    A graph $G=(V,E)$ is an \emph{$(n,d,\lambda)$-spectral expander} if $|V|=n$, the degree of $G$ is $d$, $G$ is connected, and the second-largest eigenvalue of the adjacency matrix $A(G)$ is at most $d(1-\lambda)$.

    A \emph{$(d,\lambda)$-expander family} is a sequence of graphs $(G_n)$ such that $G_n$ is a $(n,d,\lambda)$-spectral expander for each $n$, and a description of $G_n$ can be generated in time polynomial in $n$.
\end{definition}

Expander graphs have been a very fruitful line of work, and many examples of expander families with constant spectral gap are known.

\begin{theorem}[\emph{e.g.} \cite{Mar73,RVW00}]\label{thm:expander-existence}
    There are constants $d\in\N$ and $\lambda>0$ such that there exists a $(d,\lambda)$-expander family.
\end{theorem}

It will also be useful to know the expansion of the cycle graph, which is a very weak expander, but has minimal degree and nonzero expansion.

\begin{lemma}\label{lem:cycle-expansion}
    Let $C_d$ be the cycle on $d$ vertices. $C_d$ is a $(d,2,\tfrac{8}{d^2})$-spectral expander.
\end{lemma}

\begin{proof}
    Let $\omega$ be a primitive $d$-th root of unity. It is easy to see that for any $i=0,\ldots,d-1$, the vector $(1,\omega^{i},\omega^{2i},\ldots,\omega^{(d-1)i})$ is an eigenvector of the adjacency matrix $A(C_d)$ with eigenvalue $2\cos\parens*{\frac{2\pi i}{d}}$, and that these eigenvectors provide an eigenbasis. As such, the largest eigenvalue is the degree $2$ and the second-largest eigenvalue is $2\cos\parens*{\frac{2\pi}{d}}$. Thus, the expansion coefficient is
    \begin{align*}
        \lambda&=\frac{2-2\cos\parens*{\frac{2\pi}{d}}}{2}=2\sin^2\parens*{\frac{\pi}{d}}\geq\frac{8}{d^2}.\qedhere
    \end{align*}
\end{proof}

The following very nice lemma was shown in \cite{Ji2021quantum}.

\begin{lemma}[Lemma A.2 \cite{Ji2021quantum}]\label{lem:key_expand}
    Let $G=(V,E)$ be an $(n,d,\lambda)$-spectral expander, let $\mc{A}$ be a $\ast$-algebra, let $\tau$ be a state on $\mc{A}$, and let $\{a_u\}_{u\in V}$ be a collection of operators in $\mc{A}$. Then,
    \begin{align*}
        \frac{1}{n^2}\sum_{u,v}\norm{a_u-a_v}_\tau^2\leq\frac{1}{\lambda}\frac{2}{dn}\sum_{\{u,v\}\in E}\norm{a_u-a_v}_\tau^2
    \end{align*}
\end{lemma}

\begin{proof}
    Let $\tau=\rho\circ\varphi$ be the GNS representation of $\tau$, where $\varphi:\mc{A}\rightarrow\mc{B}(H)$ is a $\ast$-homomorphism and $\rho$ is a state on $\mc{B}(H)$, for some Hilbert space $H$. Let $A_u=\varphi(a_u)$. We see the adjacency matrix $A(G)$ as an operator on $\C^V$, and consider the Laplacian matrix
    $$L=I-\frac{1}{d}A(G)=\frac{1}{d}\sum_{\{u,v\}\in E}\parens*{\ket{u}-\ket{v}}\parens*{\bra{u}-\bra{v}}.$$
    Note that the $0$-eigenspace of $L$ is spanned by $\ket{\psi_0}=\frac{1}{\sqrt{n}}\sum_u\ket{u}$, and by spectral expansion $L\geq\lambda(I-\ketbra{\psi_0})$. Let $V:H\rightarrow\C^V\otimes H$ be the operator $V=\sum_{u\in V}\ket{u}\otimes A_u$. First,
    \begin{align*}
        V^\ast(L\otimes I)V=\frac{1}{d}\sum_{\{u,v\}\in E}\parens*{A_u-A_v}^\ast\parens*{A_u-A_v}.
    \end{align*}
    On the other hand,
    \begin{align*}
        V^\ast(L\otimes I)V&\geq\lambda V^\ast((I-\ketbra{\psi})\otimes I)V\\
        &=\lambda\parens*{\sum_u A_u^\ast A_u-\frac{1}{n}\sum_{u,v}A_u^\ast A_v}=\frac{\lambda}{n}\sum_{u,v}\parens*{A_u^\ast A_u-A_u^\ast A_v}\\
        &=\frac{\lambda}{2n}\sum_{u,v}\parens*{A_u-A_v}^\ast\parens*{A_u-A_v}.
    \end{align*}
    Acting by the state $\rho$ on both sides, we get
    \begin{align*}
        \frac{\lambda}{2n}\sum_{u,v}\norm{a_u-a_v}_\tau^2\leq\frac{1}{d}\sum_{\{u,v\}\in E}\norm{a_u-a_v}_\tau^2,
    \end{align*}
    giving the result.
\end{proof}

\begin{corollary}\label{cor:round-to-pvm}
    For each $u\in V$, suppose $\{A^u_i\}_{i=1}^k$ is a PVM in some finite-dimensional von Neumann algebra $\mc{M}$ such that $\frac{2}{nd}\sum_{\{u,v\}\in E}\sum_i\norm{A^u_i-A^v_i}_\tau^2\leq\varepsilon$ for some tracial state $\tau$ on $\mc{M}$. Then, there exists a PVM $\{P_i\}_{i=1}^k$ such that $\frac{1}{n}\sum_{u,i}\norm{A^u_i-P_i}_\tau^2\leq\varepsilon/\lambda$.
\end{corollary}

\begin{proof}
    By \cref{lem:key_expand} we know that $\frac{1}{n^2}\sum_{u,v}\sum_i\norm{A^u_i-A^v_i}_\tau^2\leq\frac{\varepsilon}{\lambda}$. Then, note that
    $$\sum_i\norm{A^u_i-A^v_i}_\tau^2=\sum_i\tau\parens*{A^u_i+A^v_i-2A^u_iA^v_i}=2-2\sum_i\tau(A^u_iA^v_i),$$
    and hence $\frac{1}{n^2}\sum_{u,v}\sum_i\tau(A^u_iA^v_i)\geq1-\frac{\varepsilon}{2\lambda}$. Writing the POVM $A'_i=\frac{1}{n}\sum_u A^u_i$, we have $\frac{1}{n}\sum_{u,i}\tau(A^u_iA'_i)\geq 1-\frac{\varepsilon}{2\lambda}$. As the PVMs are extremal in the convex set of POVMs in $\mc{M}$, there exists a PVM $\{P_i\}$ such that 
    $$\frac{1}{n}\sum_{u,i}\tau(A^u_iP_i)\geq\frac{1}{n}\sum_{u,i}\tau(A^u_iA'_i)\geq 1-\frac{\varepsilon}{2\lambda}.$$
    To finish the proof, note that as $\{P_i\}$ is a PVM, we have again $\sum_i\norm{A^u_i-P_i}_\tau^2=2-2\sum_i\tau(A^u_iP_i)$.
\end{proof}

To reduce the number of constraints that each variable appears in, we modify the constraint system by labelling each variable by the constraint it appears in and adding equality constraints between these labelled variables, one for each edge of an expander graph. We use the properties of the expanders to ensure equality while maintaining the low degree. We formalize this in the following definition.

\begin{definition}
    Let $G=(G_n)$ be a sequence of graphs such that $G_n=([n],E_n)$, and let $S=(X,\{(V_i,C_i)\}_{i=1}^m)$ be a $k$-ary CS. For each variable $x\in X$, let $n_x=\abs*{\set*{i}{x\in V_i}}$ and fix a bijection $r_x:\set*{i}{x\in V_i}\mapsto [n_x]$. We call the \emph{$G$-replacement of $S$} the $k$-ary CS $S|_G$ with variables $X|_G=\set*{x_i}{i\in[m],x\in V_i}$ and two types of constraints: for each $i\in[m]$, the constraint $(V_i',C_i)$ where $V_i'=\set*{x_i}{x\in V_i}$; and for each $x\in X$ and edge $\{u,v\}\in E_{n_x}$, the constraint $(V_{x,r_x^{-1}(u),r_x^{-1}(v)},C_=)$, where $V_{x,i,j}=\{x_{i},x_{j}\}$ and $C_==\set*{(a,a)}{a\in\Z_k}$ is the $k$-ary equality constraint. Given a probability distribution $\pi$ on $[m]$, define the probability distribution $\pi|_G$ via $\pi|_G(i)=\pi(i)/2$ and $\pi|_G(x,i,j)=\frac{\pi(x)}{2dn_x}$ where $\pi(x)=\sum_{i.\;x\in V_i}\frac{\pi(i)}{|V_i|}$.
\end{definition}

\begin{remark}
    If $G$ is a $(d,\lambda)$-expander family, then the left degree of the connectivity graph of $\Gamma(S|_G)$ is $d+1$; the right degree is simply the largest number of variables in a constraint, that is $\max\{2,\max_i|V_i|\}$. In particular, if $S$ is a 3SAT instance, we have that the left vertices all have degree $d+1$ and the right vertices have degree $2$ or $3$. Also, the equality constraint between two boolean variables $x,y$ can be represented by the gadget $(\lnot x\lor y\lor y)\land(x\lor\lnot y\lor \lnot y)$. If $B$ is a BCS which is a 3SAT instance, we call $B|_{G,3SAT}$ the 3SAT instance constructed by replacing all the equality constraints in $B|_G$ by the two constraints given in the above gadget. Then, the left vertices of $\Gamma(B|_{G,3SAT})$ all have degree $2d+1$ and the right vertices all have degree $3$.
\end{remark}

\begin{lemma}\label{lem:c-v-inter}
    Let $S=(X,\{(V_i,C_i)\}_{i=1}^m)$ be a CS, let $\pi$ be a probability distribution on $[m]$, and let $\tau$ be a tracial state on $\mc{A}_{c-v}(S,\pi)$. Then,
    \begin{align*}
        \defect(\tau)=\sum_{i,x\in V_i,a\in\Z_k}\frac{\pi(i)}{2|V_i|}\norm{\Pi_a(\sigma_i(x))-\Pi_a(\sigma'(x))}_\tau^2.
    \end{align*}
\end{lemma}

\begin{proof}
    In fact,
    \begin{align*}
        \defect(\tau)&=\sum_{i,x\in V_i,\phi\in C_i}\frac{\pi(i)}{|V_i|}\tau\parens*{\Phi_{V_i,\phi}(1-\Pi_{\phi(x)}(\sigma'(x)))}\\
        &=\sum_{i,x\in V_i}\frac{\pi(i)}{|V_i|}\sum_{a_y\in\Z_k.\;y\in V_i}\tau\parens[\Big]{\prod_{y\in V_i}\Pi_{a_y}(\sigma_i(y))(1-\Pi_{a_x}(\sigma'(x)))}\\
        &=\sum_{i,x\in V_i}\frac{\pi(i)}{|V_i|}\sum_{a\in\Z_k}\tau\parens[\Big]{\Pi_{a}(\sigma_i(x))(1-\Pi_{a}(\sigma'(x)))}.
    \end{align*}
    To finish the proof, note that for any projectors $P,Q$, we have that
    \begin{align*}
        \norm*{P-Q}_\tau^2=\tau(P+Q-2PQ)=\tau(P-Q)+2\tau(P(1-Q)),
    \end{align*}
    and therefore
    \begin{align*}
        \sum_{a\in\Z_k}&\frac{1}{2}\norm{\Pi_a(\sigma_i(x))-\Pi_a(\sigma'(x))}_\tau^2\\
        &=\sum_{a\in\Z_k}\frac{1}{2}\tau\parens*{\Pi_a(\sigma_i(x))-\Pi_a(\sigma'(x))}+\tau\parens*{\Pi_a(\sigma_i(x))(1-\Pi_a(\sigma'(x)))}\\
        &=\frac{1}{2}\tau(1-1)+\sum_{a\in\Z_k}\tau\parens*{\Pi_a(\sigma_i(x))(1-\Pi_a(\sigma'(x)))}\\
        &=\sum_{a\in\Z_k}\tau\parens*{\Pi_a(\sigma_i(x))(1-\Pi_a(\sigma'(x)))}.\qedhere
    \end{align*}
\end{proof}

\begin{lemma}\label{lem:eq-to-3sat}
    Consider the BCS $B_{x=y}=(\{x,y\},\{(V_1,C_1),(V_2,C_2)\})$, where $V_1=V_2=\{x,y,z\}$, and $C_1$ and $C_2$ the 3CNF constraints given by $\lnot x\lor y\lor y$ and $x\lor\lnot y\lor \lnot y$, respectively. Then, for a state $\tau$ on $\mc{A}_{c-v}(B_{x=y},\mbbm{u}_3)$, $\norm{\sigma'(x)-\sigma'(y)}_\tau^2\leq64\defect(\tau)$.
\end{lemma}

\begin{proof}
    For a variable $v$, write $v^i=\sigma_i(v)$, $v^i_a=\sigma_i(\Pi_a(v))$, $v=\sigma'(v)$, and $v_a=\sigma'(\Pi_a(v))$. Due to the constraints, we have $x^1_{-1}y^1_1=0$ and $x^2_1y^2_{-1}=0$. Note that $\norm{v^i-v}_\tau^2=4\norm{v^i_a-v_a}_\tau^2$ and hence using \cref{lem:c-v-inter},
    \begin{align*}
        \defect(\tau)&=\sum_{i,v\in V_i}\frac{1}{8|V_i|}\norm{v^i-v}_\tau^2\\
        &=\frac{1}{16}\Big(\norm{x^1-x}_\tau^2+\norm{x^2-x}_\tau^2+\norm{y^1-y}_\tau^2+\norm{y^2-y}_\tau^2\Big).
    \end{align*}
    Next,
    \begin{align*}
        \norm{x-y}_\tau^2&=\norm{1-xy}_\tau^2=4\norm{x_{-1}y_{1}+x_{1}y_{-1}}_\tau^2\\
        &=4\norm{x^1_{-1}y^1_1+x^1_{-1}(y_1-y^1_1)+(x_{-1}-x^1_{-1})y_1+x^2_{1}y^2_{-1}+x^2_{1}(y_{-1}-y^2_{-1})+(x_{1}-x^2_{1})y_{-1}}_\tau^2\\
        &=4\norm{x^1_{-1}(y_1-y^1_1)+(x_{-1}-x^1_{-1})y_1+x^2_{1}(y_{-1}-y^2_{-1})+(x_{1}-x^2_{1})y_{-1}}_\tau^2\\
        &\leq 4\Big(\norm{y-y^1}_\tau^2+\norm{y-y^2}_\tau^2+\norm{x-x^1}_\tau^2+\norm{x-x^2}_\tau^2\Big)\\
        &\leq64\defect(\tau),
    \end{align*}
    where the first inequality is  $\norm{\sum_{i=1}^ka_i}_\tau^2\leq2^{\ceil{\log k}}\sum_{i=1}^k\norm{a_i}_\tau^2$~\cite[Lemma 7.2]{MS24}.
\end{proof}

\begin{theorem}\label{thm:replacement-soundness}
    Let $S=(X,\{(V_i,C_i)\}_{i=1}^m)$ be a $k$-ary CS, and $G=(G_n)$ be a $(d,\lambda)$-expander family. For every tracial state $\tau$ on $\mc{A}_{c-v}(S|_G,\mbbm{u}_m|_G)$, there exists a tracial state $\tau'$ on $\mc{A}_{c-v}(S,\mbbm{u}_m)$ such that $\defect(\tau')\leq\frac{16L}{\lambda}\defect(\tau)$, where $L=\max_i|V_i|$. 
\end{theorem}


\begin{proof}
    Write $\pi=\mbbm{u}_m$. Write $\varepsilon_i=\frac{1}{|V_i'|}\sum_{x_i\in V_i',a}\norm{\Pi_a(\sigma_i(x_i))-\Pi_a(\sigma'(x_i))}_\tau^2$ and $$\varepsilon_x=\frac{1}{2d n_x}\sum_{\substack{i,j.\;x\in V_i\cap V_j,\\\{r_x(i),r_x(j)\}\in E_{n_x}}}\sum_{a}\norm{\Pi_a(\sigma_{x,i,j}(x_i))-\Pi_a(\sigma'(x_i))}_\tau^2+\norm{\Pi_a(\sigma_{x,i,j}(x_j))-\Pi_a(\sigma'(x_j))}_\tau^2.$$ We have $\defect(\tau)=\sum_i\frac{\pi(i)}{4}\varepsilon_i+\sum_x\frac{\pi(x)}{4}\varepsilon_x$. Since $C_{x,i,j}=C_=$ is the equality constraint between $x_i$ and $x_j$, $\sigma_{x,i,j}(x_i)=\sigma_{x,i,j}(x_j)$ 
    so that
    $$\frac{1}{d n_x}\sum_{\substack{i,j.\;x\in V_i\cap V_j,\\\{r_x(i),r_x(j)\}\in E_{n_x}}}\sum_{a}\norm{\Pi_a(\sigma'(x_i))-\Pi_a(\sigma'(x_j))}_\tau^2\leq 4\varepsilon_x.$$
    Let $\varphi:\mc{A}_{c-v}(S|_G,\mbbm{u}_{m}|_G)\rightarrow\mc{M}$ be a $\ast$-representation and $\rho:\mc{M}\rightarrow\C$ be a state such that $\tau=\rho\circ\varphi$ is the GNS representation of $\tau$. Using \Cref{cor:round-to-pvm}, there exists an order-$k$ unitary $T_x\in\mc{M}$ such that
    $$\frac{1}{n_x}\sum_{i.x\in V_i}\sum_a\norm{\varphi(\Pi_a(\sigma'(x_i)))-\Pi_a(T_x)}_\rho^2\leq\frac{4}{\lambda}\varepsilon_x.$$
    Take $\chi:\mc{A}_{c-v}(S,\mbbm{u}_m)\rightarrow\mc{M}$ to be the $\ast$-homomorphism defined via $\chi(\sigma_i(x))=\varphi(\sigma_i(x_i))$ and $\chi(\sigma'(x))=T_x$. Take $\tau'=\rho\circ\chi$. Using the fact that $\pi$ is uniform, note that $\pi(i)=\frac{1}{m}$ and $\frac{n_x}{Lm}\leq\pi(x)\leq\frac{n_x}{m}$. Then, the defect
    \begin{align*}
        \defect(\tau')&=\sum_{i}\frac{\pi(i)}{2|V_i|}\sum_{x\in V_i,a}\norm{\Pi_a(\sigma_i(x))-\Pi_a(\sigma'(x))}_{\tau'}^2\\
        &=\sum_{i}\frac{\pi(i)}{2|V_i|}\sum_{x\in V_i,a}\norm{\varphi(\Pi_a(\sigma_i(x_i)))-\Pi_a(T_x)}_\rho^2\\
        &\leq\sum_{i}\frac{\pi(i)}{|V_i|}\sum_{x\in V_i,a}\norm{\varphi(\Pi_a(\sigma_i(x_i)))-\varphi(\Pi_a(\sigma'(x_i)))}_\rho^2+\norm{\varphi(\Pi_a(\sigma'(x_i)))-\Pi_a(T_x)}_\rho^2\\
        &=\sum_i\pi(i)\varepsilon_i+\sum_x\sum_{i.\;x\in V_i}\frac{\pi(i)}{|V_i|}\sum_a\norm{\varphi(\Pi_a(\sigma'(x_i)))-\Pi_a(T_x)}_\rho^2\\
        &\leq\sum_i\pi(i)\varepsilon_i+\sum_x\frac{1}{m}\sum_{i.\;x\in V_i}\sum_a\norm{\varphi(\Pi_a(\sigma'(x_i)))-\Pi_a(T_x)}_\rho^2\\
        &\leq\sum_i\pi(i)\varepsilon_i+\frac{4}{\lambda}\sum_x\frac{n_x}{m}\varepsilon_x\\
        &\leq\sum_i\pi(i)\varepsilon_i+\frac{4L}{\lambda}\sum_x\pi(x)\varepsilon_x\leq\frac{16L}{\lambda}\defect(\tau).\qedhere
    \end{align*}
\end{proof}

\begin{proof}[Proof of \cref{thm:3sat-5}]

    Due to \cref{lem:padded-to-3sat}, there is a polynomial-time reduction from the $\exp$-padded halting problem, an $\RE$-complete problem, to $\mathrm{CSP}_{c-v}(\mathrm{3SAT})_{1,s}^\ast$ for some $s\in(0,1)$. Let $B=(X,\{(V_i,C_i)\}_{i=1}^{m})$ be a 3SAT instance and let $G=(G_n=([n],E_n))_n$ be a $(d,\lambda)$-expander family. Next, using \cref{thm:replacement-soundness}, for every trace $\tau$ on $\mc{A}_{c-v}(B|_G,\mbbm{u}_m|_G)$, there exists a trace $\tau'$ on $\mc{A}_{c-v}(B,\mbbm{u}_m)$ with defect $\defect(\tau')\leq\frac{48}{\lambda}\defect(\tau)$. We have $\mbbm{u}_m|_G(i)=\frac{1}{2m}$ and $\frac{1}{6dm}\leq\mbbm{u}_m|_G(x,i,j)\leq\frac{1}{2dm}$, hence we can replace $\mbbm{u}_m|_G$ by the uniform distribution on the $m'$ constraints of $B|_G$ while incurring only a constant penalty. Explicitly, there exists a constant $C$ such that for every trace $\tau$ on $\mc{A}_{c-v}(B|_G,\mbbm{u}_{m'})$, there exists a trace $\tau'$ on $\mc{A}_{c-v}(B,\mbbm{u}_m)$ with defect $\defect(\tau')\leq C\defect(\tau)$. Consider the family of graphs $H=(C_n)_n$ where $C_n$ is the cycle on $n$ vertices. By \cref{lem:cycle-expansion}, $C_n$ is a $(n,2,\tfrac{8}{n^2})$ expander, so $H$ is not an expander family. However, since each variable in $B|_G$ is contained in at most $d+1$ constraints, $H$ behaves like an expander family with $\lambda=\frac{8}{(d+1)^2}$ with respect to $B|_G$. As such, using \cref{thm:replacement-soundness} again, for every trace $\tau$ on $\mc{A}_{c-v}(B|_G|_H,\mbbm{u}_{m'}|_H)$, there exists a trace $\tau'$ on $\mc{A}_{c-v}(B,\mbbm{u}_m)$ with defect $\defect(\tau')\leq6(d+1)^2C\defect(\tau)$. To finish, note that $B|_G|_H$ has 3CNF or equality constraints. Using \cref{lem:eq-to-3sat}, we can replace all the equality constraints by gadgets built out of 3CNF constraints. As such, there exists a 3SAT-5 instance $B'$, a probability distribution and a constant $C'$ such that for every trace $\tau$ on $\mc{A}_{c-v}(B',\pi')$, there exists a trace $\tau'$ on $\mc{A}_{c-v}(B,\pi)$ such that $\defect(\tau')\leq C'\defect(\tau)$. 
\end{proof}

\bibliographystyle{alpha}
\bibliography{bib}
\begin{appendices}

\appendix

\section{Proof of weak expansion}\label{sec:weak-expansion}

\begin{lemma}[Weak-Expansion \cite{guruswami-bypassing}]\label{lem:weak-expansion}
    Let $\Gamma$ be a left- and right-regular bipartite graph. Define $G=(V,E)$ to be the graph on the right vertices of $\Gamma$ with an edge between $x$ and $y$ for each left vertex they are both connected to in $\Gamma$. Then, for any $\delta > 0$ and vertex subset $V'\subseteq V$ such that $|V'| = \delta\cdot |V|$, the number of edges between the vertices in $V'$ is at least $\delta^2 |E|$.
\end{lemma}

\begin{proof}
    To simplify notation, let us denote the set of left vertices $U$, the left degree as $d_U$, the right degree as $d_V$, and the neighbourhood of $u\in U$ in $\Gamma$ as $N_u\subseteq V$. For some $u\in U$, let $p_u$ be the fraction of neighbours of $U$ that are in $V'$, \emph{i.e.} 
    \begin{align*}
        p_u=\frac{\abs{N_u \cap V'}}{\abs{N_u}} = \frac{\abs{N_u \cap V'}}{d_U}.
    \end{align*}
    Now the number of edges from $V'$ to $U$ is given by $\abs{V'} d_V$. On the other hand, the number of edges from $V'$ to $U$ can also be written as $d_U \sum_{u\in U} p_u$. Thus we have $\sum_{u\in U} p_u =\abs{V'} d_V / d_U$. Therefore we have
    \begin{align*}
        \expect_u p_u = \frac{1}{\abs{U}} \sum_u p_u = \frac{\abs{V'} d_V }{\abs{U} d_U } = \frac{\delta \abs{V} d_V }{\abs{U} d_U } =\delta,
    \end{align*}
    where we used the assumption that $|V'| = \delta\cdot |V|$ and that $\abs{V} d_V=\abs{U} d_U$.

    We note that by construction the total number of edges in $G$ is given by $\abs{E}= \sum_u \abs{N_u}^2 = \abs{U} d_U^2$ as we put an edge between the neighbours of every $u\in U$, including self-loops. Similarly the number of edges between vertices in $V'$ is given by $\abs{E'}=\sum_u \abs{N_u \cap V'}^2$. Then we have
    \begin{align*}
        \frac{\abs{E'}}{\abs{E}} = \frac{\sum_u \abs{N_u \cap V'}^2}{\abs{U} d_U^2} = \expect_u p_u^2 \geq (\expect_u p_u)^2 =\delta^2,
    \end{align*}
    where we used Jensen's inequality.
\end{proof}

\section{Conditional Linear Distributions and Uniform Marginals}\label{sec:uniformmarginals}

In this section, we prove that there is a reduction from the halting problem to a BCS-$\MIP^*$ protocol with \emph{uniform marginals}. The uniform marginals are another way to guarantee a reduction from the halting problem to \cref{thm:replacement-soundness}. Before we begin, we require several definitions from \cite{Ji2021quantum}.

\begin{definition}[\cite{ji_mip_re}]
    Let $\mbb{F}$ be a finite field and $V$ be a vector space over $\mbb{F}$. We define conditional linear functions recursively as follows. A \emph{level $1$ conditional linear function} on $V$ is a linear map $L:V\rightarrow V$. Now let $k>1$ and suppose $V$ has the direct sum decomposition $V=V_1\oplus V_2\oplus\ldots\oplus V_k$. A \emph{level $k$ conditional linear function} is a map $L:V\rightarrow V$ such that $$L(x_1\oplus\ldots\oplus x_k)=L_1(x_1)\oplus L_{>1,L_1(x_1)}(x_2\oplus\ldots\oplus x_k),$$ where $L_1:V_1\rightarrow V_1$ is a linear map and for each $u\in V_1$, $L_{>1,u}:V_2\oplus\ldots\oplus V_k\rightarrow V_2\oplus\ldots\oplus V_k$ is a level $k-1$ conditional linear map.
\end{definition}

A conditional linear map $L$ can be fully described by a direct sum decomposition $V=V_1\oplus\ldots\oplus V_k$ and a set of linear maps $L_1,L_{2,u_1},L_{3,u_1,u_2},\ldots$, where
$$L(x_1\oplus\ldots\oplus x_k)=L_1(x_1)\oplus L_{2,L_1(x_1)}(x_2)\oplus L_{3,L_1(x_1),L_{2,L_1(x_1)}(x_2)}(x_3)\oplus\ldots.$$

\begin{definition}[\cite{ji_mip_re}]
    Let $V$ be a vector space over a finite field $\mbb{F}$. A \emph{conditional linear distribution} on $V\times V$ is a probability distribution $\pi$ on $V\times V$ such that there exist conditionally linear functions $L^A,L^B:V\rightarrow V$ such that $$\pi(x,y)=\frac{1}{|V|}\abs*{\!\set*{z\in V}{(x,y)=(L^A(z),L^B(z))}}.$$
    To specify the conditional linear functions, we write $\pi=\pi_{L^A,L^B}$
\end{definition}

We say that a conditional linear distribution $\pi_{L^A,L^B}$ is \emph{efficiently sampleable} if $L^A$ and $L^B$ are efficiently computable.

\begin{lemma}\label{lem:computable}
    Let $\pi=\pi_{L^A,L^B}$ be an efficiently sampleable conditional linear distribution. Then, $\pi(x,y)$ is efficiently computable, on input $(x,y)\in V\times V$. Similarly, on input $x\in V$, $\pi_A(x)=\sum_y\pi(x,y)$ and $\pi_B(x)=\sum_y\pi(y,x)$ are efficiently computable.
\end{lemma}

\begin{proof}
    Let $V=V_1\oplus\ldots\oplus V_k$ and $V=W_1\oplus\ldots\oplus W_l$ be the direct sum decompositions corresponding to $L^A$ and $L^B$, respectively, and $L^A_1,L^A_{2,u_1},\ldots$ and $L^B_1,L^B_{2,u_1},\ldots$ be the linear functions. First, we show how to compute the pre-image $(L^A)^{-1}(x)$. Decompose $x=x_1\oplus\ldots\oplus x_k$. Then, for $z=z_1\oplus\ldots\oplus z_k$, if $x=L^A(z)$, then $L^A_1(z_1)=x_1$, $L^A_{2,x_1}(z_2)=x_2$, $L^A_{3,x_1,x_2}(z_3)=x_3$, and so on. For $i=1,\ldots,k$, let $U_{i}^A=\set*{z_{i}\in V_i}{L^A_{i,x_1,x_2,\ldots,x_{i-1}}(z_{i})=x_i}$. These are affine linear subspaces, so efficient to compute and describe. Hence $(L^A)^{-1}(x)=U^A_1\oplus\ldots\oplus U^A_k$ is efficiently computable. Similarly,  $(L^B)^{-1}(x)=U^B_1\oplus\ldots\oplus U^B_l$, where $U_{i}^B=\set*{z_{i}\in W_i}{L^B_{i,y_1,y_2,\ldots,y_{i-1}}(z_{i})=y_i}$. As such, $$\pi(x,y)=\frac{1}{|V|}\abs*{(U^A_1\oplus\ldots\oplus U^A_k)\cap(U^B_1\oplus\ldots\oplus U^B_l)}.$$
    $(U^A_1\oplus\ldots\oplus U^A_k)\cap(U^B_1\oplus\ldots\oplus U^B_l)$ is the intersection of two affine subspaces, so it can be described as the solution space of a linear system. As such, it can be efficiently computed, so the number of elements can be computed efficiently.

    Further, since $\pi_A(x)=\frac{1}{|V|}\abs*{(L^A)^{-1}(x)}=\frac{1}{|V|}\abs*{U^A_1\oplus\ldots\oplus U^A_k}$ and $\pi_B(y)=\frac{1}{|V|}\abs*{U^B_1\oplus\ldots\oplus U^B_l}$, these are also efficiently computable.
\end{proof}

\begin{definition}[\cite{ji_mip_re}]\label{def:t_type}
    Let $V$ be a vector space over a finite field $\mbb{F}$ and $T$ be a finite set (called the set of types). A \emph{$T$-typed conditional linear distribution} on $V\times V$ is a probability distribution $\pi$ on $(T\times V)\times(T\times V)$ such that there exists a graph $G=(T,E)$ and conditional linear functions $L^A_t,L^B_t:V\rightarrow V$ for all $t\in T$ such that
    \begin{align*}
        \pi((t,x),(u,y))=\begin{cases}\frac{1}{|E_\rightarrow||V|}\abs*{\!\set*{z\in V}{(x,y)=(L^A_t(z),L^B_u(z))}}&\{t,u\}\in E\\0&\text{ otherwise}\end{cases},
    \end{align*}
    where $E_\rightarrow=\set*{(t,u)}{\{t,u\}\in E}$ is the induced direct edge set. To specify the graph and conditional linear functions, we write $\pi=\pi_{G,\{L^A_t\},\{L^B_t\}}$.
\end{definition}

We remark that graphs with self-loops are allowed in \cref{def:t_type}.

\begin{lemma}\label{lem:typed-computable}
    Let $\pi=\pi_{G,\{L^A_t\},\{L^B_t\}}$ be an efficiently sampleable $T$-typed conditional linear distribution (where $T$ is constant size). Then, on input $((t,x),(u,y))\in(T\times V)^2$, $\pi((t,x),(u,y))$, $\pi_A(t,x)$, and $\pi_B(u,y)$ are efficiently computable. 
\end{lemma}

\begin{proof}
    The lemma follows by applying \cref{lem:computable} to the conditional linear distribution $\pi_{L^A_t,L^B_u}$.
\end{proof}

\begin{definition}
    Let $\pi=\pi_{L^A,L^B}$ be a conditional linear distribution. We define the \emph{oracularization} of $\pi$ as the $T$-typed conditional linear distribution $\pi_{G,\{L^A_t\},\{L^B_t\}}$ where $T=\{A,B,O\}$, $G=(T,E)$ is the graph with $E=\set*{\{O,O\},\{A,O\},\{B,O\}}$, and $L^A_A=L^B_A=L^A$, $L^A_B=L^B_B=L^B$, and $L^A_O=L^B_O=\Id$.
\end{definition}

This definition combines the notions of oracularization from \cite{ji_mip_re} and \cite{MS24}. From the former, it follows the model of conditional linear distributions, while the question distribution is the same as the latter, with slightly different weighting. In particular, synchronous games give rise to BCS games under this oracularization.


For a $T$-typed conditional linear distribution $\pi=\pi_{G,\{L^A_t\},\{L^B_t\}}$, note that
\begin{align*}
    \pi_A((t,x))=\frac{N_G(t)}{|E_\rightarrow||V|}|(L_t^A)^{-1}(x)|,
\end{align*}
and identically $\pi_B((t,x))=\frac{N_G(t)}{|E_\rightarrow||V|}|(L_t^B)^{-1}(x)|$

\begin{definition}
    Let $\pi=\pi_{G,\{L^A_t\},\{L^B_t\}}$ be a $T$-typed conditional linear distribution. For $(t,x)\in T\times V$, define $N^A_{t,x}=N_G(t)|(L_t^A)^{-1}(x)|$ and $N^B_{t,x}=N_G(t)|(L_t^B)^{-1}(x)|$. The \emph{padding} of $\pi$ is the distribution $\pi^{pad}$ on $X\times Y$, where $X=\bigcup_{t\in T,x\in V}\{(t,x)\}\times[N^A_{t,x}]$ and $Y=\bigcup_{t\in T,x\in V}\{(t,x)\}\times[N^B_{t,x}]$ such that $\pi^{pad}((t,x,m),(u,y,n))=\frac{\pi((t,x),(u,y))}{N^A_{t,x}N^B_{u,y}}$
\end{definition}

The important property of the padded distribution $\pi^{pad}$ is that it has uniform marginals. In fact,
\begin{align*}
    \pi^{pad}_A((t,x,n))=\sum_{n=1}^{N^B_{u,y}}\frac{\pi_A((t,x))}{N^A_{t,x}N^B_{u,y}}=\frac{1}{|E_\rightarrow||V|},
\end{align*}
and identically $\pi^{pad}_B((u,y,n))=\frac{1}{|E_\rightarrow||V|}$.

\begin{lemma}\label{lem:padded-efficient}
    Suppose that $\pi=\pi_{G,\{L^A_t\},\{L^B_t\}}$ is an efficiently sampleable $T$-typed conditional linear distribution. Then, $\pi^{pad}$ is efficiently sampleable.
\end{lemma}

\begin{proof}
    To show this, we demonstrate how to sample from $\pi^{pad}$. Since $\pi$ is efficiently sampleable, first sample $((t,x),(u,y))$ according to $\pi$. Now, due to \cref{lem:typed-computable}, we can efficiently compute $\pi_A((t,x))$ and $\pi_B((u,y))$. But, $N_{t,x}^A=|E_\rightarrow||V|\pi_A((t,x))$ and $N_{u,y}^B=|E_\rightarrow||V|\pi_B((u,y))$. Then, sample $m$ uniformly from $[N_{t,x}^A]$ and $n$ uniformly from $[N_{u,y}^B]$, and output $((t,x,m),(u,y,n))$. By construction, this is distributed according to $\pi^{pad}$.
\end{proof}

\begin{lemma}\label{lem:padded-value}
    Let $(\pi,V)$ be a nonlocal game where $\pi=\pi_{G,\{L^A_t\},\{L^B_t\}}$ is a $T$-typed conditional linear distribution. Define the padded verifier $V^{pad}(a,b|(t,x,n),(u,y,m))=V(a,b|(t,x),(u,y))$. Then, the values of $(\pi,V)$ and $(\pi^{pad},V^{pad})$ are equal over any convex set of correlations.
\end{lemma}

\begin{proof}
    Let $p=\{p(a,b|(t,x),(t,y))\}$ be a correlation for $(\pi,V)$. Define the correlation $q(a,b|(t,x,m),(u,y,n))=p(a,b|(t,x),(t,y))$ for $(\pi^{pad},V^{pad})$. Then, the value
    \begin{align*}
        w(q)&=\sum_{a,b,t,x,m,u,y,n}\pi^{pad}((t,x,m),(u,y,n))V^{pad}(a,b|(t,x,m),(u,y,n))q(a,b|(t,x,m),(u,y,n))\\
        &=\sum_{a,b,t,x,m,u,y,n}\frac{\pi((t,x),(u,y))}{|N^A_{t,x}||N^B_{u,y}|}V(a,b|(t,x),(u,y))p(a,b|(t,x),(u,y))\\
        &=\sum_{a,b,t,x,u,y}\pi((t,x),(u,y))V(a,b|(t,x),(u,y))p(a,b|(t,x),(u,y))=w(p).
    \end{align*}
    For the converse, let $q=\{q(a,b|(t,x,m),(u,y,n))\}$ be a correlation for $(\pi^{pad},V^{pad})$. Define the correlation $$p(a,b|(t,x),(u,y))=\frac{1}{N^A_{t,x}N^B_{u,y}}\sum_{m,n}q((a,b|(t,x,m),(u,y,n))$$ for $(\pi,V)$. Then,
    \begin{align*}
        w(p)&=\sum_{a,b,t,x,u,y}\pi((t,x),(u,y))V(a,b|(t,x),(u,y))p(a,b|(t,x),(u,y))\\
        &=\sum_{a,b,t,x,m,u,y,n}\frac{\pi((t,x),(u,y))}{|N^A_{t,x}||N^B_{u,y}|}V(a,b|(t,x),(u,y))q(a,b|(t,x,m),(u,y,n))\\
        &=\sum_{a,b,t,x,m,u,y,n}\pi^{pad}((t,x,m),(u,y,n))V^{pad}(a,b|(t,x,m),(u,y,n))q(a,b|(t,x,m),(u,y,n))\\
        &=w(q).\qedhere
    \end{align*}
\end{proof}

\begin{definition}
    Write $\BCSMIP^\ast_{unif}(q,a)_{c,s}$ for the class of languages decidable by one-round constraint-constraint $\BCSMIP^\ast$ protocols with question length $q(n)$, answer length $a(n)$, completeness $c$, soundness $s$, such that the question distribution is symmetric with uniform marginals.
\end{definition}

Remark that a constraint-constraint game with symmetric question distribution $\pi(x,y)$ that has uniform marginals induces a constraint-variable game with uniform constraint distribution $\pi'(x)=\sum_y\pi(x,y)$. Further, the gap $1-s$ is only reduced under this transformation by a factor that is linear in the answer length (constant for the protocols we consider),

 \begin{theorem}\label{thm:uniformmarginals}
     There exists $s\in(0,1)$ such that there is a polynomial-time reduction from the halting problem to $\BCSMIP^\ast_{unif}(\polylog(n),O(1))_{1,s}$.
 \end{theorem}

\begin{proof}[Proof of \cref{thm:uniformmarginals}]
    By \cite{ji_mip_re,dong2023computational,fu2025succinct}, there exists a synchronous\\ $\MIP^\ast(\polylog(n),O(1))_{1,1/2}$ protocol for the halting problem where the question distribution is an efficiently sampleable conditional linear distribution. By oracularizing, this gives rise to a $\BCSMIP^\ast(\polylog(n),O(1))_{1,s}$ protocol for the halting problem where the question distribution is an efficiently sampleable $T$-typed conditional linear distribution, for some $s\in(0,1)$. Now, using \cref{lem:padded-efficient,lem:padded-value} we can replace the question distribution by its padded distribution, while preserving efficient sampling and quantum game value. But the padded distribution has uniform marginals, and hence this gives a $\BCSMIP^\ast_{unif}(\polylog(n),O(1))_{1,s}$ protocol for the halting problem.
\end{proof}

\begin{corollary}
    There exists $s\in(0,1)$, such that there is a polynomial-time reduction from the halting problem to the problem of deciding whether the winning probability of the constraint-variable game of a succinctly-presented $\SAT$ instance with uniform probability distribution on the constraints is $1$ or $<s$.
\end{corollary}

The proof follows by using the classical homomorphism and subdivision transformations of~\cite{MS24,culfmastel24} to reduce an arbitrary $\BCSMIP_{unif}^\ast$ instance to a succinct $\SAT$ instance presented as a constraint-constraint game with uniform marginals; and then using the $12$-homomorphism to the corresponding constraint-variable game~\cite[Lemma 3.7]{culfmastel24} to get the wanted uniform distribution.
\end{appendices}
\end{document}